\documentclass[sigconf]{acmart}

\usepackage{booktabs} 

\setcopyright{rightsretained}
\usepackage{epsfig}
\usepackage{graphicx}
\usepackage{balance}  
\usepackage{amsmath}
\usepackage{amssymb}
\usepackage{subfig}
\usepackage{caption}
\usepackage{booktabs}
\usepackage{multirow}
\usepackage{epstopdf}
\usepackage{algorithmicx, algpseudocode}
\usepackage[linesnumbered,ruled,vlined]{algorithm2e}
\usepackage{url}
\usepackage{color}
\usepackage{capt-of}
\usepackage[overload]{textcase}
\usepackage{enumitem}
\usepackage{hyperref}

\usepackage{array}
\newcolumntype{L}[1]{>{\raggedright\let\newline\\\arraybackslash\hspace{0pt}}m{#1}}
\newcolumntype{C}[1]{>{\centering\let\newline\\\arraybackslash\hspace{0pt}}m{#1}}
\newcolumntype{R}[1]{>{\raggedleft\let\newline\\\arraybackslash\hspace{0pt}}m{#1}}

\newtheorem{Lemma}{Lemma}
\newtheorem{Definition}{Definition}

\newcommand\I{\mathbb{I}}
\newcommand\E{\mathbb{E}}

\newcommand\R{\mathcal{R}}

\newcommand{\Cov}{\mathrm{Cov}}

\def\OPT{\textsf{OPT}}

\def\src{\textsf{src}}

\def\NRI{\textsf{nSI}}
\def\ERI{\textsf{eSI}}

\def\V{\mathcal{V}}
\def\G{\mathcal{G}}
\def\E{\mathcal{E}}
\def\DN{\mathbb{D}^{(n)}}
\def\DE{\mathbb{D}}

\def\D{\mathbb{D}}

\def\NIS{\textsf{nSIA}}
\def\EIS{\textsf{eSIA}}
\def\OPTe{\textsf{OPT}_k}
\def\OPTn{\textsf{OPT}^{(n)}_k}
\def\saw{\textsf{SAW}}
\def\hsaw{\textsf{HSAW}}

\acmDOI{10.475/123_4}

\acmISBN{123-4567-24-567/08/06}

\acmConference[KDD'17]{ACM SIGKDD conference}{Aug. 2012}{Halifax, Nova Scotia, Candada} 
\acmYear{2017}
\copyrightyear{2017}

\acmPrice{15.00}

\setcopyright{none}
\renewcommand\footnotetextcopyrightpermission[1]{} 

\begin{document}
	
	\title{Blocking Self-avoiding Walks Stops Cyber-epidemics: \\A Scalable GPU-based Approach}
	
	\author{Hung Nguyen, Alberto Cano, and Thang Dinh}
	\affiliation{%
		\institution{Virginia Commonwealth University}
		\city{Richmond} 
		\state{VA} 
		\postcode{23220}
	}
	\email{{hungnt, acano, tndinh}@vcu.edu}
	
	\author{Tam Vu}
	\affiliation{%
		\institution{University of Colorado, Denver}
		\city{Denver} 
		\state{CO} 
		\postcode{80204}
	}
	\email{tam.vu@ucdenver.edu }

	\renewcommand{\shortauthors}{Nguyen et al.}
	\begin{abstract}
	Cyber-epidemics, the widespread of fake news or propaganda through social media, can cause devastating economic and political consequences. A common countermeasure against cyber-epidemics is to disable a small subset of suspected social connections or accounts to effectively contain the epidemics. An example is the recent shutdown of 125,000 ISIS-related Twitter accounts. Despite many proposed methods to identify such subset, none are scalable enough to provide high-quality solutions in nowadays billion-size networks.
        
	To this end, we investigate the Spread Interdiction problems that seek most effective links (or nodes) for removal under the well-known Linear Threshold model. We propose novel CPU-GPU methods that \emph{scale to networks with billions of edges}, yet, possess \emph{rigorous theoretical guarantee} on the solution quality. At the core of our methods is an \emph{$O(1)$-space} out-of-core algorithm to generate a new type of random walks, called \textit{Hitting Self-avoiding Walks} (\hsaw{}s). Such a low memory requirement enables handling of big networks and, more importantly, hiding latency via scheduling of millions of threads on GPUs. Comprehensive experiments on real-world networks show that our algorithms provides much higher quality solutions and are several order of magnitude faster than the state-of-the art. Comparing to the (single-core) CPU counterpart, our GPU implementations achieve significant speedup factors  up to 177x on a single GPU and 338x on a GPU pair.
	\end{abstract}

	%
	%

	%
	%
	
	%
	%
	
	\vspace{-0.1in}
		\vspace{-0.1in}
	\maketitle
	
	\vspace{-0.19in}
\section{Introduction}
\label{sec:intro}
Cyber-epidemics have caused significant economical and political consequences, and even more so in the future due to the increasing popularity of social networks.
Such widespread of fake news and propaganda has potential to pose serious threats to global security. For example, through social media, terrorists have recruited thousands of supporters who have carried terror acts including bombings in the US, Europe, killing dozens of thousands of innocents, and created worldwide anxiety \cite{Bombing}. The rumor of explosions at the White House injuring President Obama  caused \$136.5 billion loss in stock market \cite{White} or the recent burst of fake news has significantly influenced the 2016 election \cite{Allcott17}. 

To contain those cyber-epidemics, one common strategy is to disable user accounts or social connects that could potentially be vessels for rumor propagation through the ``word-of-mouth'' effect. For example, Twitter has deleted 125,000 accounts linked to terrorism\cite{Twitter1} since the middle of 2015 and U.S. officials have called for shutting down al-shababs Twitter accounts\cite{Twitter2}. Obviously, removing too many accounts/links will negatively affect legitimate experience, possibly hindering the freedom of speech. Thus it is critical to identify small subsets of social links/user accounts whose removal effectively contains the epidemics. 

Given a social network, which can be abstracted as a graph in which nodes represent users and edges represent their social connections, the above task is equivalent to the problem of identifying nodes and edges in the graph to remove such that it minimizes the (expected) spread of the rumors under a  diffusion model. In a ``blind interdiction'' manner, \cite{Tong12, Dinh15} investigate the problem when no information on the sources of the rumors are available. Given the infected source nodes, Kimura et al. \cite{Kimura08} and \cite{Kuhlman13} proposed heuristics to remove edges to minimize the spread from the sources. Remarkably, Khalil et al. \cite{Khalil14} propose the first $1-1/e-\epsilon$-approximation algorithm for the edges removal problem under the linear threshold model \cite{Kempe03}. However, the number of samples needed to provide the theoretical guarantee is too high for practical purpose. Nevertheless, none of the proposed methods can scale to large networks with billions of edges and nodes.

In this paper, we formulate and investigate two Spread Interdiction problems, namely  \emph{Edge-based Spread Interdiction} (\ERI{}) and \emph{Node-based Spread Interdiction} (\NRI{}). The problems consider a graph, representing a social network and a subset of suspected nodes that might be infected with the rumor. They seek for a size-$k$ set of edges (or nodes) that removal minimize the spread from the suspected nodes under the  well-known linear threshold (LT) model \cite{Kempe03}. Our major contribution is the two hybrid GPU-based algorithms, called  \EIS{}  and \NIS{}, that possess distinguished characteristics:
\begin{itemize}
	\item \textbf{Scalability}: Thanks to the highly efficient self-avoiding random walks generation on GPU(s), our algorithms runs \emph{several order of magnitude faster} than its CPU's counterpart as well as the state-of-the-art method\cite{Khalil14}. The proposed methods take only seconds on networks with billions of edges and can work on even bigger networks via stretching the data across multiple GPUs.
	\item \textbf{Riorous quality guarantee}: Through extensive analysis, we show that our methods return $(1-1/e-\epsilon)$-approximation solutions w.h.p. Importantly, our methods can effectively determine a minimal number of \hsaw{} samples to achieve the theoretical guarantee for given $\epsilon > 0$. In practice, our solutions are consistently 10\%-20\% more effective than the runner up when comparing to the centrality-based, influence-maximization-based methods, and the current state of the art in \cite{Khalil14}.
\end{itemize}


The foundation of our proposed methods is a theoretical connection between Spread Interdiction and a new type of random walks, called \textit{Hitting Self-avoiding Walks} (\hsaw{}s). The connection allows us to find
 the most effective edges (nodes) for removal through finding those who appear most frequently on the \hsaw s. The bottle neck of this approach is, however, the generation of \hsaw s, which requires repeatedly generation of self-avoiding walks until one reach a suspected node. Additionally, the standard approach to generate self-avoiding walks requires $\Omega(n)$ space per thread to store whether each node has been visited. This severely limits the number of threads that can be launched concurrently.
 
To tackle this challenge, we propose a novel \emph{$O(1)$-space} out-of-core algorithm to generate \hsaw. Such a low memory requirement enables handling of big networks on GPU and, more importantly, hiding latency via scheduling of millions of threads.  Comparing to the (single-core) CPU counterpart, our GPU implementations achieve  significant speedup factors  up to 177x on a single GPU and 388x on a GPU pair, making them several order of magnitude faster than the state-of-the art method \cite{Khalil14}.


Our contributions are summarized as follows:
\begin{itemize}
	\item We formulate the problem of stopping the cyber-epidemics by removing nodes and edges as two interdiction problems and establish an important \emph{connection between the Spread Interdiction problems and blocking Hitting Self-avoiding Walk} (\hsaw).
	\item We propose out-of-core $O(1)-space$ \hsaw{} sampling algorithm that allows concurrent execution of millions of threads on GPUs.
		For big graphs that do not fit into a single GPU, we also provide distributed algorithms on multiple GPUS via the techniques of graph partitioning and node replicating.
	 Our sampling algorithm might be of particular interest for those who are into sketching influence dynamics of billion-scale networks.
	\item Two 	$(1-1/e-\epsilon)$-approximation algorithms, namely,  \EIS{} and \NIS{}, for the edge and node versions of the spread interdiction problems. Our approaches bring together rigorous theoretical guarantees and practical efficiency.
	\item We conduct comprehensive experiments on real-world networks with up to 1.5 billion edges. The results suggest the superiority of our methods in terms of solution quality (10\%-20\% improvement) and running time (2-3 orders of magnitude faster).
\end{itemize}

\textbf{Organization}. We present the LT model and formulate two Spread Interdiction problems on edges and nodes in Section~\ref{sec:model}. Section~\ref{sec:saw} introduces Hitting Self-avoiding Walk (\hsaw) and proves the monotonicity and submodularity, followed by \hsaw{} sampling algorithm in Section~\ref{sec:sam} with parallel and distributed implementations on GPUs. The complete approximation algorithms are presented in Section~\ref{sec:ssa}. Lastly, we present our experimental results in Section~\ref{sec:exp}, related work in Section~\ref{sec:relatedwork} and conclusion in Section~\ref{sec:con}.
	
\section{Models and Problem Definitions}
\label{sec:model}
We consider a social network represented by a directed probabilistic graph $\mathcal{G}=(V, E, w)$  that contains $|V|=n$ nodes and $|E|=m$ weighted edges. Each edge $(u, v) \in E$ is associated with an infection weight $w(u, v) \in [0, 1]$ which indicates the likelihood that $u$ will infect $v$ once $u$ gets infected. 

Assume that we observe in the network a set of suspected nodes $V_I$ that might be infected with misinformation or viruses. However, we do not know which ones are actually infected. Instead, the probability that a node $v \in V_I$ is given by a number $p(v) \in [0, 1]$. In a social network like Twitter, this probability can be obtained through analyzing tweets' content to determine the likelihood of misinformation being spread. By the same token, in computer networks, remote scanning methods can be deployed to estimate the probability that a computer gets infected by a virus.

The Spread Interdiction problems aim at selecting a set of nodes or edges whose removal results in maximum influence suspension of the infected nodes. We assume  a subset $C$ of candidate nodes (or edges) that we can remove from the graph. The $C$ can be determined depending on the situation at hand. For example, $C$ can contains (highly suspicious) nodes from $V_I$ or even nodes outside of $V_I$, if we wish to contains the rumor rapidly. Similarly, $C$ can contains edges that are incident to suspected nodes in $V_I$ or $C=E$ if we wish to maximize the effect of the containment.

We consider that the infection spreads according to  the well-known \textit{Linear Threshold (LT)} diffusion model~\cite{Kempe03}.

\subsection{Linear Threshold  Model}
\label{subsec:model}
In the LT model, each user $v$ selects an activation threshold $\theta_v$ \emph{uniformly random} from $[0, 1]$. The edges' weights must satisfy a condition that, for each node, the sum of all in-coming edges' weights is at most 1, i.e., $\sum_{u \in V} w(u, v) \leq 1, \forall v\in V$. The diffusion happens in discrete time steps $t=0, 1, 2, \ldots, n$. At time $t=0$, a set of users $S \subseteq V$,called the \emph{seed set}, are infected and all other nodes are not. We also call the infected nodes \emph{active}, and uninfected nodes \emph{inactive}. An inactive node $v$ at time $t$ becomes active at time $t+1$ if $\sum_{\text{ active neighbors $u$ of } v  } w(u, v) \geq \theta_v$. The infection spreads until no more nodes become active.

Given $\mathcal G=(V, E, w)$ and a seed set $S \subset V$, the \emph{influence spread} (or simply \emph{spread}) of $S$, denoted by $\I_{\mathcal{G}}(S)$, is the expected number of infected nodes at the end of the diffusion process. Here the expectation is taken over the randomness of all thresholds $\theta_v$.

One of the extensively studied problem is the influence maximization problem \cite{Kempe03}. The problem asks for a seed set $S$ of $k$ nodes to maximize $\I_{\mathcal{G}}(S)$. In contrast, this paper considers the case when the seed set (or the distribution over the seed set) is given and aims at identifying a few edges/nodes whose removals effectively reduce the influence spread.

\textbf{LT live-edge model.} In \cite{Kempe03}, the LT model is shown to be equivalent to the \emph{live-edge} model where each node $v\in V$ picks at most one incoming edge with a probability equal to the edge weight. Specifically, a \emph{sample graph} $G=(V, E' \subset E)$ is generated from $\mathcal G$ according to the following rules: 1) for each node $v$, at most one incoming edge is selected; 2) the probability of selecting edge $(u,v)$ is $w(u,v)$ and there is no incoming edge to $v$ with probability $(1-\sum_{u \in N^-(v)}w(u,v))$. 

Then the influence spread $\I_{\mathcal{G}}(S)$ is equal the expected number of nodes reachable from $S$ in sample graph $G$, i.e.,
\vspace{-0.04in}
\begin{align*}
	\I_{\mathcal{G}}(S) &= \sum_{G \sim \mathcal G} \{\text{\#nodes in $G$ reachable from }S\} \Pr[G \sim \G],
\end{align*}
\vspace{-0.1in}

\noindent where $G \sim \mathcal G$ denotes the sample graph $G$ induced from the stochastic graph $\mathcal G$ according to the live-edge model.

For a sample graph $G$ and a node $v\in V$, define
\newcommand{\twopartdef}[3]
{
	\left\{
	\begin{array}{ll}
		#1 & \mbox{if } #2 \\
		#3 & \mbox{otherwise}
	\end{array}
	\right.
}
\vspace{-0.05in}
\begin{align}
	\chi^G(S,v) = \twopartdef {1} {v \text{ is reachable from } S \text{ in } G} {0}
\end{align}
\vspace{-0.15in}

We can rewrite the equation for influence spread as
\begin{align}
	\I_{\mathcal{G}}(S) &= \sum_{v \in V} \sum_{G \sim \mathcal G} \chi^G(S,v) \Pr[G \sim \G]
	\label{eq:compute_inf} 	= \sum_{v \in V} \I_{\G}(S, v),
\end{align}
\vspace{-0.1in}

\noindent where $\I_{\G}(S, v)$ denotes the probability that node $v \in V$ is eventually infected by the seed set $S$.

\textbf{Learning Parameters from Real-world Traces.}
Determining the infection weights in the diffusion models is itself a hard problem and have been studied in various researches \cite{Goyal10,Cha10}. In practice, this infection weight $w(u,v)$ between nodes $u$ and $v$ is usually estimated by the interaction frequency from $u$ to $v$ \cite{Kempe03,Tang15} or learned from additional sources, e.g., action logs \cite{Goyal10}.


\vspace{-0.1in}
\subsection{Spread Interdiction in Networks}
Denote by $\V_I = (V_I, p)$, the set of suspected nodes $\V_I$ and their probabilities of being the sources. $\V_I$ defines a probability distribution over possible seed sets. The probability of a particular seed set $X \subseteq V_I$ is given by
\vspace{-0.05in}
\begin{align}
	\label{eq:prob_src_set}
	\Pr[X \sim \V_I] = \prod_{u \in X} p(u)\prod_{v \in V_I\backslash X} (1-p(v)).
\end{align}
\vspace{-0.1in}

By considering all possible seed sets $X\sim \V_I$, we further define the expected influence spread of $\V_I$ as follows,
\vspace{-0.05in}
\begin{align}
	\label{def:eis}
	\I_{\G}(\V_I) = \sum_{X \sim \V_I} \I_{\G}(X)\Pr[X \sim \V_I].
\end{align}
\vspace{-0.1in}

We aim to remove $k$ nodes/edges from the network to minimize the spread of infection  from the suspected nodes  $\V_I$ (defined in Eq.~\ref{def:eis}) in the residual network. Equivalently, the main goal in our formulations is to find a subset of edges (node) $T$ that maximize the \textit{influence suspension} defined as,
\vspace{-0.05in}
\begin{align}
	\label{eq:sus_nod}
	\D(T, \V_I) = \I_{\G}(\V_I) - \I_{\G'}(\V_I),
\end{align}
where $\G'$ is the residual network obtained from $\G$ by removing edges (nodes) in $S$. When $S$ is a set of nodes, all the edges adjacent to nodes in $S$ are also removed from the $\G$. 

We formulate the two interdiction problems as follows.
\vspace{-0.05in}
\begin{Definition}[Edge-based Spread Interdiction (\ERI{})]
	\label{def:edge_remove}
	Given $\mathcal{G} = (V,E,w)$, $\V_I$, a set of suspected nodes and their probabilities of being infected $\V_I = (V_I,p)$, a candidate set $C \subseteq E$ and a budget $1\leq k \leq |C|$, the \ERI{} problem asks for a $k$-edge set $\hat T_k \subseteq E$ that maximizes the influence suspension $\DE(T_k,\V_I)$.
	\begin{align}
	\label{eq:edge_rev}
	\hat T_k = {\arg\max}_{T_k \subseteq C, |T_k| = k}\DE(T_k,\V_I),
	\end{align}
	\vspace{-0.2in}
	
	\noindent where,
	\vspace{-0.1in}
	\begin{align}
	\label{eq:de}
	\DE(T_k,\V_I) = \I_{\mathcal{G}}(\V_I) - \I_{\mathcal{G}'}(\V_I).
	\end{align}
	\vspace{-0.2in}
\end{Definition}
\vspace{-0.05in}
\vspace{-0.05in}
\begin{Definition} [Node-based Spread Interdiction (\NRI{})]
	\label{def:node_remove}
	Given a stochastic graph $\mathcal{G} = (V,E,w)$, a set of suspected nodes and their probabilities of being infected $\V_I = (V_I,p)$, a candidate set $C \subseteq V$ and a budget $1 \leq k \leq |C|$, the \NRI{} problem asks for a $k$-node set $\hat S_k$ that maximizes the influence suspension $\DN(S_k,\V_I)$.
	\vspace{-0.05in}
	\begin{align}
		\label{eq:node_rev}
			\hat S_k = {\arg\max}_{S_k \subseteq C, |S_k| = k}\DN(S_k,\V_I),
	\end{align}
	\vspace{-0.2in}
	
	\noindent where
	\vspace{-0.1in}
	\begin{align}
		\label{eq:dn}
		\DN(S_k,\V_I) = \I_{\mathcal{G}}(\V_I) - \I_{\mathcal{G}'}(\V_I).
	\end{align}
	\vspace{-0.2in}
	
	\noindent is the influence suspension of $S_k$ as defined in Eq.~\ref{eq:sus_nod}.
\end{Definition}

We also abbreviate $\DN(S_k,\V_I)$ by  $\DN(S_k)$ and $\DE(T_k,\V_I)$ by $\DE(T_k)$ when the context is clear.


\vspace{0.05in}

\textbf{Complexity and Hardness}. The hardness results of \NRI{} and \ERI{} problems are stated in the following theorem.
\begin{theorem}
	\label{theo:hard}
	\NRI{} and \ERI{} are NP-hard and cannot be approximated within $1-1/e-o(1)$ under $P \neq NP$.
\end{theorem}
The proof is in our appendix. In the above definitions, the suspected nodes in $\V_I = (V_I,p)$ can be inferred from the frequency of suspicious behaviors or their closeness to known threats. These probabilities are also affected by the seriousness of the threats.

\textbf{Extension to Cost-aware Model.} One can generalize the \NRI{} and \ERI{} problems to replace the set of candidate $C$ with an assignment of removal costs for edges (nodes).  This can be done by incorporating the cost-aware version of max-coverage problem in  \cite{Khuller99, Nguyen16}. For the shake of clarity, we, however, opt for the uniform cost version in this paper.

\section{Hitting Self-avoiding Walks}
\label{sec:saw}
In this section, we first introduce a new type  of \textit{Self-avoiding Walk} (\saw{}), called \textit{Hitting \saw{}} (\hsaw{}), under the LT model. These \hsaw{}s and how to generate them are keys of our proofs and algorithms. Specifically, we prove that the Spread Interdiction problems are equivalent to identifying the ``most frequently edges/nodes'' among a collection of \hsaw{}s.

\vspace{-0.1in}
\subsection{Definition and Properties}
First, we define \textit{Hitting Self-avoiding Walk} (\hsaw{}) for a sample graph $G$ of $\G$.
\begin{Definition} [Hitting Self-avoiding Walk (\hsaw{})]
	\label{def:self_walk}
	Given a sample graph $G = (V,E)$ of a stochastic graph $\G=(V,E,w)$ under the live-edge model (for LT), a sample set $X \subseteq \V_I$, a walk $h = <v_1,v_2,\dots,v_l>$ is called a hitting self-avoiding walk if $\forall i \in [1,l-1], (v_{i+1},v_i)\in E, \forall i,j \in [1,l], v_i \neq v_j$ and $h \cap X = \{v_l\}$.
\end{Definition}

An \hsaw{} $h$ starts from a node $v_1$, called the source of $h$ and denoted by $\src(h)$, and consecutively walks to an incoming neighboring node without visiting any node more than once. From the definition, the  distribution of \hsaw{}s depends on the distribution of the sample graphs $G$, drawn from $\G$ following the live-edge model, the distributions of the infection sources $\V_I$, and  $\src(h)$.

According to the live-edge model (for LT), each node has at most one incoming edge. This leads to three important properties.

\textit{Self-avoiding}: An \hsaw{} has no duplicated nodes. Otherwise there is a loop in $h$ and at least one of the node on the loop (that cannot be $v_l$) will have at least two incoming edges, contradicting the live-edge model for LT.

\textit{Walk Uniqueness}: Given a sample graph $G \sim \G$ and $X\sim \V_I$, for any node $v \in V\backslash X$, there is at most one \hsaw{} $h$ that starts at node $v$. To see this, we can trace from $v$ until reaching a node in $X$. As there is at most one incoming edge per node, the trace is unique.

\textit{Walk Probability}: Given a stochastic graph $\G = (V,E,w)$ and  $\V_I$, the probability of having a particular \hsaw{} $h = <v_1, v_2, \dots, v_l>$, where $v_l\in \V_I$, is computed as follows,
	\vspace{-0.05in}
	\begin{align}
		\label{eq:compute_path_inf}
		\Pr[h \in \G] = p(v_l) \prod_{u \in V_I \cap h, u\neq v_l}(1-p(u)) \sum_{G\sim \G} \Pr[G \sim \G] \cdot\mathbf 1_{h \in G}\nonumber 
	\end{align}
	\vspace{-0.15in}
	\begin{align}
		\text{ \ }= p(v_l) \prod_{u \in V_I \cap h, u\neq v_l}(1-p(u))\prod_{i = 1}^{l-1} w(v_{i+1},v_i),
	\end{align}
where $h \in \G$ if all the edges in $h$ appear in a random sample $G \sim \G$.

Thus, based on the properties of \hsaw{}, we can define a probability space $\Omega_h$ which has the set of elements being all possible \hsaw{}s and the probability of a \hsaw{} computed from Eq.\ref{eq:compute_path_inf}.

\vspace{-0.1in}
\subsection{Spread Interdiction $\leftrightarrow$ \hsaw{} Blocking}
\label{sec:mon_sub}
From the probability space of \hsaw{} in a stochastic network $\G$ and set $\V_I$ of infected nodes, we prove the following important connection between influence suspenion of a set of edges and a random \hsaw{}. We say $T \subseteq E$ \emph{interdicts} an \hsaw{} $h_j$ if $T \cap h_j \neq \emptyset$. When $T$ interdicts $h_j$, removing $T$ will disrupt $h_j$, leaving $\src(h_j)$ uninfected.
\vspace{-0.05in}
\begin{theorem}
	\label{lem:rr_edge}
	Given a graph $\mathcal{G}=(V, E, w)$ and a set $\V_I$, for any random \hsaw{} $h_j$ and any set $T \in E$ of edges, we have
	\vspace{-0.05in}
	\begin{align}
		\label{eq:lem_edge}
		\DE(T,\V_I) = \I_\G(\V_I) \Pr[T \text{ interdicts } h_j].
	\end{align}
\end{theorem}

The proof is presented in the extended version in \cite{Extended}. Theorem~\ref{lem:rr_edge} states that the influence suspension of a set $T \subseteq E$ is proportional to the probability that $T$ intersects with a random \hsaw{}.  Thus, to maximize the influence suspension, we find a set of edges that hits the most \hsaw{}s. This motivates our sampling approach:
\begin{itemize}
	\item[(1)] Sample $\theta$ random \hsaw{}s to build an estimator the influence suspensions of many edge sets,
	\item[(2)] Apply \textsf{Greedy} algorithm over the set of \hsaw{} samples to find a solution $\hat T_k$ that blocks the most \hsaw{} samples.
\end{itemize}
The challenges in this approach are how to efficiently generate random \hsaw{}s and what the value of $\theta$ is to provide guarantee.

As a corollary of Theorem~\ref{lem:rr_edge}, we obtain the monotonicity and submodularity of the influence suspension function $\DE(T_k, \V_I)$.
\vspace{-0.1in}
\begin{corollary}
	\label{lem:mon_sub}
	The influence suspension function $\DE(S)$ where $T$ is the set of edges, under the LT model is monotone,
	\begin{align}
		\forall T \subseteq T', \DE(T) \leq \DE(T'),
	\end{align}
	and submodular, i.e. for any $(u,v) \notin T'$,
	\begin{align}
		\DE(T \cup \{v\}) - \DE(T) \geq \DE(T' \cup \{(u,v)\}) - \DE(T').
	\end{align}
\end{corollary}
The proof is presented in the extended version \cite{Extended}. The monotonicity and submodularity indicates that the above greedy approach will return $(1-1/e-\epsilon)$ approximation solutions,  where $\epsilon>0$ depends on the number of generated \hsaw. To provide a good guarantee, a large number of \hsaw{} are needed, making generating \hsaw{} the bottleneck of this approach.

\section{Scalable HSAW Sampling Algorithm}
We propose our sampling algorithm to generate \hsaw{} samples on massive parallel GPU platform. We begin with the simple CPU-based version of the algorithm.
\label{sec:sam}
\begin{algorithm} \small
	\caption{\textsf{\hsaw{} Sampling Algorithm}}
	\label{alg:ris}
	\KwIn{Graph $\G$, suspect set $\V_I$ and $p(v), \forall v \in \V_I$}
	\KwOut{A random \hsaw{} sample $h_j$}
	\While{$\emph{\textsf{True}}$}{
		Pick a node $v$ uniformly at random;\\
		Initialize $h_j = \emptyset$;\\
		\While{$\emph{\textsf{True}}$}{
			$h_j = h_j \cup \{(u,v)\}$ ($h_j = h_j \cup \{u\}$ for node version);\\
			Use \emph{live-edge} model to select an edge $(u,v) \in E$;\\
			\If{no edge selected}{\textbf{break};}
			\If{edge $(u,v)$ is selected}{
				\If{$u \in V_I$ \emph{\textbf{and}} $\textsf{rand}() \leq p(u)$}{ \textbf{return} $h_j$;}
				\If{$u \in h_j$}{\textbf{break};}
				Set $v = u$;\\
			}
		}
	}
\end{algorithm}

\vspace{-0.05in}
\subsection{CPU-based Algorithm to Generate \hsaw s}
\label{subsec:saw_sam}

Algorithm~\ref{alg:ris} describes our \hsaw{} sampling procedure which is based on the \emph{live-edge} model in Section~\ref{sec:model}. 

The algorithm follows a \textit{rejection sampling scheme} which repeatedly generate random \saw{} (Lines~2-14) until getting a \hsaw{} (Line~1). The \saw{} sampling picks a random node $v$ and  follows the live-edge model to select an incoming edge to $(u, v)$ (Lines~5-14). Then it replaces $v$ with $u$ and repeat the process until either: 1) no live-edge is selected (Lines~7-8) indicating that $h_j$ does not reach to an infected node; 2) $h_j$ hits to a node in $\V_I$ and that node is actually an infected node (Lines~10-11) or 3) edge $(u,v)$ is selected but $u$ closes a cycle in $h_j$. Only in the second case, the algorithm terminates and return the found \hsaw.

\begin{figure}[!ht]
	\vspace{0.05in}
	\centering
	\includegraphics[width=0.65\linewidth]{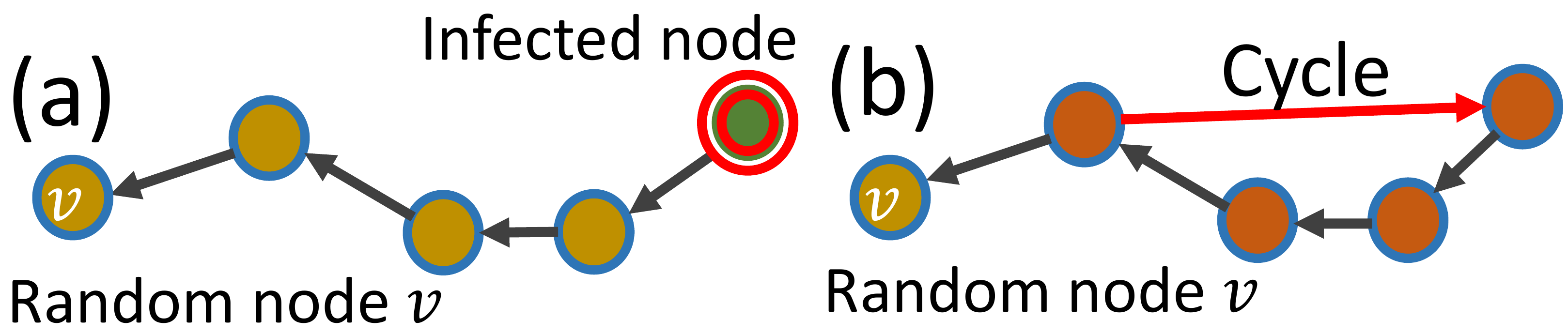}
	\vspace{-0.15in}
	\caption{\small (a) no cycle, \hsaw{} found  (b) a single cycle, no \hsaw.}
	\label{fig:path}
\end{figure}
 The algorithm is illustrated in Fig.~\ref{fig:path}. In (\textbf{a}), the simple path travels through several nodes and reach an infected node. In (\textbf{b}), the algorithm detects a cycle. 


\vspace{-0.15in}
\subsection{Parallel \hsaw{} Generation on GPU\lowercase{(s)}}
\label{sec:gpu}

GPUs with the massive parallel computing power offer an attractive solution for generating \hsaw,  the major bottleneck. 
As shown in the previous subsection, generating a \hsaw{} requires repeatedly generating \saw s. 
Since the \saw{} samples are independent, if we can run millions of \saw{} generation threads on GPUs,  we can maximize the utility of GPUs' cores and minimize the stalls due to pipelines hazard or memory accesses, i.e., minimize \textit{latency hiding}. Moreover, only the hitting \saw{} need to be transported back to CPU, thus the GPU-CPU communication is minimal.

\textbf{Challenges.}
Due to the special design of GPU with a massive number of parallel threads, in the ideal case, we can speed up our algorithms vastly if memory accesses are \textit{coalesced} and there is no \textit{warp divergence}. However, designing such algorithms to fully utilize GPUs requires attention to the GPU architecture.


Moreover, executing millions of parallel threads means each thread has little memory to use. Unfortunately, the CPU-based Algorithm to generate \hsaw (Alg.~\ref{alg:ris}) can use up to $\Omega(n)$ space to track which nodes have been visited. For large networks, there is not enough memory to launch a large number of threads.

%

We tackle the above challenges and design a new lightweight \hsaw{} generation algorithm. Our algorithm, presented in Alg.~\ref{alg:par_thread}, requires only $O(1)$ space per thread. Thus, millions of threads can be invoked concurrently to maximize the efficiency. The algorithm \textsf{ThreadSample} in Alg.~\ref{alg:par_thread} consists of three efficient techniques: $O(1)$-space Path-encoding, $O(1)$-space Infinite Cycle Detection and Sliding Window Early Termination to generate \hsaw.


\vspace{-0.1in}
\subsubsection{Memory-efficient Path-encoding}
The first technique $O(1)$-space Path-encoding aims at generating \saw{} samples on GPU cores using only constant memory space. We take advantage of a typical feature of modern pseudo-random number generators 
that a random number is generated by a function with the input (seed) being the random number generated in the previous round,
\vspace{-0.05in}
\begin{align}
	\label{eq:rand}
	r_i = f(r_{i-1}) \qquad (i \geq 1)
\end{align}
\vspace{-0.17in}

\noindent where $r_0$ is the initial seed that can be set by users. Those  generators are based on linear recurrences and proven in \cite{Vigna} to be extremely fast and passing strong statistical tests. 

Thus, if we know the value of the random seed at the beginning of the \saw{} generator and the number of traversal steps, we can reconstruct the whole walks. As a result, the \saw{} sampling algorithm only needs to store the set of initial random seeds and the walk lengths. 
The Alg.~\ref{alg:par_thread}  is similar Alg.~\ref{alg:ris} except it does not return a \saw{} but only two numbers $\textsl{Seed}_h$ and $\textsl{Len}_h$ that encode the walk.

To detect cycle (line 17), \textsf{ThreadSample} use the following two heuristics to detect most of the cycles. As the two heuristics can produce false negative (but not false positive), there is small chance that \textsf{ThreadSample} will return some walks with cycles. However, the final checking of cycle in Alg.~\ref{alg:par_saw} will make sure only valid \hsaw{} will be returned.
\vspace{-0.1in}
\subsubsection{Infinite Cycle Detection}
To detect cycle in \saw{} sampling (line~17 in Alg.~\ref{alg:par_thread}), we adopt two constant space Cycle-detection algorithms: the Floyd's \cite{Knuth98} and Brent's algorithms \cite{Knuth98}.

The Floyd's algorithm only requires space for two pointers to track a generating path. These pointers move at different speeds, i.e., one is twice as fast as the other. Floyd's guarantees to detect the cycle in the first traversing round of the slower pointer on the cycle and in the second round of the faster one.
The Floyd's algorithm maintains two sampling paths pointed by two pointers and thus, needs two identical streams of random live-edge selections. 

Differently, the Brent's algorithm cuts half of the computation of Floyd's algorithm by requiring a single stream of live-edge selections. The algorithm compares the node at position $2^{i-1}, i \ge 1$ with each subsequent sequence value up to the next power of two and stops when it finds a match. In our experiments, this algorithm is up to 40\% faster than Floyd's. Overall, both of the algorithms only need memory for two pointers and have the complexity of $O(|h_j|)$. 
The Brent's algorithm combined with cycle detection results in a speedup factor of 51x in average compared to a single CPU core.
%

\vspace{-0.1in}
\subsubsection{Short Cycle Detection with Cuckoo Filter}
Many cycles if exist often have small size. Thus, we use Cuckoo filter \cite{Fan14} of a small fixed size $k$ to index and detect the cycle among the last $k$ visited nodes. 
Our experimental results (with $k=2$) show that this short cycle detection improves further other acceleration techniques to a speedup factor of 139x.

\let\oldnl\nl
\newcommand{\nonl}{\renewcommand{\nl}{\let\nl\oldnl}}
\begin{algorithm} \small
	\caption{\textsf{ThreadSample} - Sampling on a GPU thread}
	\label{alg:par_thread}
	\KwIn{$l$ and $\textsl{ThreadID}$}
	Global pool $H$ of \hsaw{} samples; \\
	Initialize a generator $\textsc{PRG}.\textsl{Seed} \leftarrow \textsc{PRG2}(\textsl{ThreadID})$;\\
	\For{i = 1 \emph{\textbf{ to }} 8}{\textsc{PRG}.\textsl{next()}\tcp*{Burn-in period}}
	\For{i = 1 \emph{\textbf{to}} l}{
		$\textsl{Seed}_h = \textsc{PRG}.\textsl{next}(); \textsl{Len}_h = 0$;\\
		Use \textsc{PRG} to pick a node $v$ uniformly at random;\\
		\While{$\emph{\textsf{True}}$}{
			Use \textsc{PRG} to select an edge $(u,v) \in E$ following the \emph{live-edge} LT model;\\
			\If{no edge selected \emph{\textbf{ or }} $\textsl{Len}_h \geq n$}{\textbf{break};}
			\If{edge $(u,v)$ is selected}{
				\If{$u \in V_I$}{
					Use \textsc{PRG} with probability $p(u)$: \\
						\Indp $H = H \text{ }\cup <\textsl{Seed}_h, \textsl{Len}_h + 1>$;\\
						\textbf{break};
				}
				\If{$\text{cycle detected at }u$}{\textbf{break};}
				Set $v = u; \textsl{Len}_h = \textsl{Len}_h + 1$;\\
			}
		}
	}
\end{algorithm}
\begin{algorithm}\small
	\caption{Parallel \hsaw{} Sampling Algorithm on GPU}
	\label{alg:par_saw}
	\KwIn{Graph $\G$, $\V_I$, $p(v), \forall v \in \V_I$}
	\KwOut{$\R$ - A stream of \hsaw{} samples}
	$i = 0;$\\
	\While{\emph{\textsf{True}}}{
		Initialize global $H = \emptyset, l = 10; \textsf{thread}_{\max}$ depends on GPU model;\\
		Call \textsf{ThreadSample}(\textsl{ThreadID}, $l$) $\forall \textsl{ThreadID} = 1..\textsf{thread}_{\max}$;\\
		\ForEach{$<\textsl{Seed}_h, \textsl{Len}_h> \in H$}{
			Reconstruct $h$ from $\textsl{Seed}_h, \textsl{Len}_h$;\\
			\If{ $h$\text{ has no cycle}}{
				$i \leftarrow i + 1$;\\
				Add $R_i = \{\text{edges in } h\}$ to stream $\R$;
			}
		}
	}
\end{algorithm}
%

\vspace{-0.08in}
\subsubsection{Combined Algorithm} The combined algorithm of generating \hsaw{} on GPU is presented in Alg.~\ref{alg:par_saw} which generates a stream of \hsaw{} samples $h_1, h_2, \dots$. The main component is a loop of multiple iterations. Each iteration calls $\textsf{thread}_{\max}$, i.e. maximum number of threads in the GPU used, threads to execute Alg.~\ref{alg:par_thread} that runs on a GPU core and generates at most $l$ \hsaw{} samples. Those samples are encoded by only two numbers $\textsl{Seed}_h$ and $\textsl{Len}_h$ which denotes the starting seed of the random number generator and the length of that \hsaw{}. Based on these two numbers, we can reconstruct the whole \hsaw{} and recheck the occurrence of cycle. If no cycles detected, a new \hsaw{} $R_i = \{\text{edges in } h\}$ is added to the stream. The small parameter $l$ prevents thread divergence.


Recall that Alg.~\ref{alg:par_thread} is similar to that of Alg.~\ref{alg:ris} except:
\begin{itemize}
	\vspace{-0.051in}
	\item[1)] It only stores three numbers: node $v$, $\textsl{Seed}_h$ and $\textsl{Len}_h$.
	\item[2)] It uses two random number generator $\textsc{PRG}$ and $\textsc{PRG2}$ which are in the same class of linear recurrence (Eq.~\ref{eq:rand}). $\textsf{PRG}$ goes through the burn-in period to gurantee the randomness (Lines~3-4).
	\item[3)] Cycle detection in Line~17 can be Floyd's, Brent's or one with Cuckoo Filter (this requires rechecking in Alg.~\ref{alg:par_saw}).
\end{itemize}
Thus, the algorithms requires only a constant space and has the same time complexity as \hsaw{} sampling in Alg.~\ref{alg:ris}.

\vspace{-0.1in}
\subsection{Distributed Algorithm on Multiple GPUs}
In case the graph cannot fit into the memory of a single GPU,  we will need to distribute the graph data across multiple GPUs. We refer to this approach as \emph{Distributed algorithm on Multiple GPUs}.

We use the folklore approach of partitioning the graph into smaller (potentially overlapping) partitions. Ideally, we aim at partitioning the graph that minimizes the inter-GPU communication. This is equivalent to minimizing the chance of a \hsaw{} crossing different partitions. To do this, we first apply the standardized METIS \cite{Lasalle15} graph partitioning techniques into $p$ partitions where $p$ is the number of GPUs. Each GPU will then receive a partition and generate samples from that subgraph. The number of samples generated by each GPU is proportional to the number of nodes in the received partition. 
We further reduce the crossing \hsaw{} by extending each partition to include nodes that are few hop(s) away. The number of hops away is called \textit{extension parameter}, denoted by $h$. We use $h=1$ and $h=2$ in our experiments.

\section{Approximation Algorithms}
\label{sec:ssa}
This section focuses on the question of detecting the minimal number of \hsaw{} to guarantee $(1-1/e-\epsilon)$ approximation and the complete present of \EIS. We adopt the recent Stop-and-Stare framework \cite{Nguyen163} proven to be efficient, i.e. meeting theoretical lower bounds on the number of samples.
\begin{algorithm} \small
	\caption{\textsf{Greedy} algorithm for maximum coverage}
	\label{alg:max-cover}
	\KwIn{A set $\R_t$ of \hsaw{} samples, $C \subseteq E$ and $k$.}
	\KwOut{An $(1 - 1/e)$-optimal solution $\hat T_k$ on samples.}
	$\hat T_k = \emptyset$;\\
	\For{$i = 1 \emph{\textbf{ to }} k$}{
		$\hat e \leftarrow \arg \max_{e \in C \backslash \hat T_k}(\Cov_{\R_t}(\hat T_k\cup \{e\}) - \Cov_{\R_t}(\hat T_k))$;\\
		Add $\hat c$ to $\hat T_k$;\\
	}
	\textbf{return} $\hat T_k$;
\end{algorithm}
\begin{algorithm} \small
	\caption{\textsf{Check} algorithm for confidence level}
	\label{alg:check}
	\KwIn{$\hat T_k, \R_t, \R'_t, \epsilon, \delta$ and $t$.}
	\KwOut{\textsf{True} if the solution $\hat T_k$ meets the requirement.}
	Compute $\Lambda_1$ by Eq.~\ref{eq:lambda_1};\\
	\If{$\Cov_{\R'_t}(\hat T_k) \geq \Lambda_1$}{
		$\epsilon_1 = \Cov_{\R_t}(\hat T_k)/\Cov_{\R'_t}(\hat T_k) - 1$; \\
		$\epsilon_2 = \epsilon \sqrt{\frac{|\R'_t|(1+\epsilon)}{2^{t-1} \Cov_{\R'_t}(\hat T_k)}}$; $\epsilon_3 = \epsilon \sqrt{\frac{|\R'_t|(1+\epsilon) (1-1/e-\epsilon)}{(1+\epsilon/3)2^{t-1} \Cov_{\R'_t}(\hat T_k)}}$; \\
		$\epsilon_t = (\epsilon_1 + \epsilon_2 + \epsilon_1 \epsilon_2)(1-1/e-\epsilon) + (1-1/e)\epsilon_3$; \\
		\If{$\epsilon_t \leq \epsilon$}{\textbf{return} \textsf{True};}
	}
	\textbf{return} \textsf{False};
\end{algorithm}

\begin{algorithm} \small
	\caption{Edge Spread Interdiction Algorithm (\EIS{})}
	\label{alg:eis}
	\KwIn{Graph $\G$, $\V_I$, $p(v), \forall v \in \V_I$, $k$, $C \subseteq E$ and $0 \leq \epsilon, \delta \leq 1$.}
	\KwOut{$\hat T_k$ - An $(1 - 1/e - \epsilon)$-near-optimal solution.}
	Compute $\Lambda$ (Eq.~\ref{eq:lambda}), $N_{\max}$ (Eq.~\ref{eq:guard}); $t = 0$; \\
	A stream of \hsaw{} $h_1, h_2, \dots$ is generated by Alg.~\ref{alg:par_saw} on GPU;\\
	\Repeat{$|\R_t|\geq N_{\max}$}{
		$t = t+1$; $\R_t = \{ R_1, \dots, R_{\Lambda 2^{t-1}} \}; \R'_t = \{ R_{\Lambda 2^{t-1} + 1}, \dots, R_{\Lambda 2^{t}} \}$;\\
		$\hat T_k \leftarrow \textsf{Greedy}(\R_t, C, k)$;\\
		\If{$\emph{\textsf{Check}}(\hat T_k, \R_t, \R'_t, \epsilon,\delta) = \emph{\textsf{True}}$}{\textbf{return} $\hat T_k$;}
	}
	\textbf{return} $\hat T_k$;\\
\end{algorithm}

\vspace{-0.1in}
\subsection{Edge-based Spread Interdiction Algorithm}
Similar to \cite{Tang14,Tang15,Nguyen163}, we first derive a threshold 
\vspace{-0.05in}
\begin{align}
	\label{eq:theta}
	\theta = (2-\frac{1}{e})^2(2+\frac{2}{3}\epsilon) \hat \I_\G(\V_I) \cdot \frac{\ln (6/\delta)+\ln {m \choose k}}{\OPTe{} \epsilon^2}.
\end{align}
Using $\theta$ \hsaw{} samples, the greedy algorithm (Alg.~\ref{alg:max-cover}) guarantees to returns a $(1-1/e-\epsilon)$ approximate solution with a probability at least $1-\delta/3$ (Theorem~\ref{theo:app} in the extended version \cite{Extended}.) 

Unfortunately, we cannot compute this threshold directly as it involves two unknowns $\hat \I_\G(\V_I)$ and $\OPTe{}$. The Stop-and-Stare framework in \cite{Nguyen163} untangles this problem by utilizing two independent sets of samples: one for finding the candidate solution using \textsf{Greedy} algorithm and the second for out-of-sample verification of the candidate solution's quality. This strategy guarantees to find a $1-1/e-\epsilon$ approximation solution within at most a (constant time) of any theoretical lower bounds such as the above $\theta$ (w.h.p.)
%

\textbf{\EIS{} Algorithm.} The complete algorithm \EIS{} is presented in Alg.~\ref{alg:eis}. It has two sub-procedures: \textsf{Greedy}, Alg.~\ref{alg:max-cover},  and \textsf{Check}, Alg.~\ref{alg:check}.

\textsf{Greedy}: Alg.~\ref{alg:max-cover}  selects a candidate solution $\hat T_k$ from a set of \hsaw{} samples $\R_t$. This implements the greedy scheme that iteratively selects from the set of candidate edges $C$ an edge that maximizes the marginal gain. The algorithm stops after selecting $k$ edges.

\textsf{Check}:  Alg.~\ref{alg:check}  verifies if the candidate solution $\hat T_k$  satisfies the given precision error $\epsilon$. It computes the error bound provided in the current iteration of \EIS{}, i.e. $\epsilon_t$ from $\epsilon_1, \epsilon_2, \epsilon_3$ (Lines~4-6), and compares that with the input $\epsilon$. This algorithm consists of a checking condition (Line~2) that examines the coverage of $\hat T_k$ on the independent set $\R'_t$ of \hsaw{} samples with $\Lambda_1$,
	\vspace{-0.05in}
	\begin{align}
		\label{eq:lambda_1}
		\Lambda_1 = 1+(1+\epsilon)(2+\frac{2}{3}\epsilon)\ln (\frac{3t_{\max}}{\delta}) \frac{1}{\epsilon^2},
	\end{align}
	\vspace{-0.15in}
	
	\noindent where $t_{\max} = \log_2\left( \frac{2N_{\max}}{(2+2/3 \epsilon) \ln (\delta/3) 1/\epsilon^2} \right)$ is the maximum number of iterations run by \EIS{} in Alg.~\ref{alg:eis} (bounded by $O(\log_2 n)$). The computations of $\epsilon_1, \epsilon_2, \epsilon_3$ are to guarantee the estimation quality of $\hat T_k$ and the optimal solution $T^*_k$.

The main algorithm in Alg.~\ref{alg:eis} first computes the upper-bound on neccessary \hsaw{} samples $N_{\max}$ i.e.,
\vspace{-0.05in}
\begin{align}
	\label{eq:guard}
	N_{\max} = (2-\frac{1}{e})^2(2+\frac{2}{3}\epsilon) m \cdot \frac{\ln (6/\delta)+\ln {m \choose k}}{k \epsilon^2},
\end{align}
\vspace{-0.15in}

\noindent and $\Lambda$ i.e.,
\vspace{-0.1in}
\begin{align}
	\label{eq:lambda}
	\Lambda = (2+\frac{2}{3}\epsilon)\ln (\frac{3 t_{\max}}{\delta}) \frac{1}{\epsilon^2},
\end{align}
\vspace{-0.15in}

\noindent Then, it enters a loop of at most $t_{\max}=O(\log n)$ iterations. In each iteration, \EIS{} uses the set $\R_t$ of first $\Lambda 2^{t-1}$  \hsaw{} samples to find a candidate solution $\hat T_k$ by the \textsf{Greedy} algorithm (Alg.~\ref{alg:max-cover}). Afterwards, it checks the quality of $\hat T_k$ by the \textsf{Check} procedure (Alg.~\ref{alg:check}). If the \textsf{Check} returns \textsf{True} meaning that $\hat T_k$ meets the error requirement $\epsilon$ with high probability, $\hat T_k$ is returned as the final solution. 

In cases when \textsf{Check} algorithm fails to verify the candidate solution $\hat T_k$ after $t_{\max}$ iteations, \EIS{} will be terminated by the guarding condition $|\R_t| \geq N_{\max}$ (Line~9).

\textbf{Optimal Guarantee Analysis.}
We prove that \EIS{} returns an $(1-1/e-\epsilon)$-approximate solution for the \ERI{} problem with probability at least $1-\delta$ where $\epsilon \in (0, 1)$ and  $1/\delta = O(n)$ are the given precision parameters.
\begin{theorem}
	\label{theo:app}
	Given a graph $\G=(\V,\E,w)$, a probabilistic set $\V_I$ of suspected nodes,  candidate edge set $C \subseteq E$, $0 < \epsilon < 1$, $1/\delta = O(n)$  as the precision parameters and budget $k$, \emph{\EIS{}} returns an $(1-1/e-\epsilon)$-approximate solution $\hat T_k$ with probability at least $1-\delta$,
	\begin{align}
	\Pr[\DE(\hat T_k) \geq (1-1/e-\epsilon)\emph{\textsf{OPT}}_k^{(e)}] \geq 1-\delta.
	\end{align}
\end{theorem}

\vspace{-0.05in}
\textbf{Comparison to Edge Deletion in \cite{Khalil14}}.
The recent  work in \cite{Khalil14} selects $k$ edges to maximize the \textit{sum} of influence suspensions of nodes in $\V_I$ while our \ERI{} problem considers $\V_I$ as a whole and maximize the influence suspension of $\V_I$. The formulation in \cite{Khalil14} reflects the case where only a single node in $\V_I$ is the seed of propaganda and each node has the same chance. In contrast, \ERI{} considers a more practical situations in which each node $v$ in $\V_I$ can be a seed independently with probability $p(v)$. Such condition is commonly found when the propaganda have been active for some time until triggering the detection system. In fact, the method in \cite{Khalil14} can be applied in our problem and vise versa. However, \cite{Khalil14} requires an impractically large number of samples to deliver the $(1-1/e-\epsilon)$ guarantee. 

\vspace{-0.1in}
\subsection{Node-based Spread Interdiction Algorithm}
Similar to Theorem~\ref{lem:rr_edge}, we can also establish the connection between identifying nodes for removal and identifying nodes that appear frequently in \hsaw s.
\begin{theorem}
	\label{lem:rr_node}
	Given $\mathcal{G}=(V, E, w)$, a random \hsaw{} sample $h_j$ and a probabilistic set $\V_I$, for any set $S \in V$,
	\begin{align}
		\DN(S,\V_I) = \I_\G(\V_I) \Pr[S \text{ interdicts } h_j]. \nonumber
	\end{align}
\end{theorem}
Thus, the \NIS{} algorithm for selecting $k$ nodes to remove and maximize the influence suspension is similar to \EIS{} except:
\begin{itemize}
	\item[1)] The \textsf{Greedy} algorithm selects nodes with maximum marginal gains into the candidate solution $\hat S_k$.
	\item[2)] The maximum \hsaw{} samples is cumputed as follows,
	\vspace{-0.05in}
	\begin{align}
		N_{\max} = (2-\frac{1}{e})^2(2+\frac{2}{3}\epsilon)n \cdot \frac{\ln (6/\delta)+\ln {n \choose k}}{k\epsilon^2}.
	\end{align}
%
\end{itemize}
The approximation guarantee is stated below.
\begin{theorem}
	\label{theo:sam}
	Given a graph $\G=(\V,\E,p)$, a probabilistic set $\V_I$ of possible seeds with their probabilities, $C \subseteq V$, $0 < \epsilon < 1$, $1/\delta = O(n)$ and a budget $k$, \emph{\NIS{}} returns an $(1-1/e-\epsilon)$-approximate solution $\hat S_k$ with probability at least $1-\delta$,
	\begin{align}
	\Pr[\DN(\hat S_k) \geq (1-1/e-\epsilon) \emph{\textsf{OPT}}^{(n)}_k] \geq 1-\delta,
	\end{align}
	where $S^*_k$ is an optimal solution of $k$ nodes.
\end{theorem}
Both the complete algorithm for node-based Spread Interdiction and the proof of Theo.~\ref{theo:sam} is presented in our extended version~\cite{Extended}.

\section{Experimental Evaluation}
\label{sec:exp}

\begin{figure*}[!ht]
	\vspace{-0.1in}
	\subfloat[Pokec]{
		\includegraphics[width=0.24\linewidth]{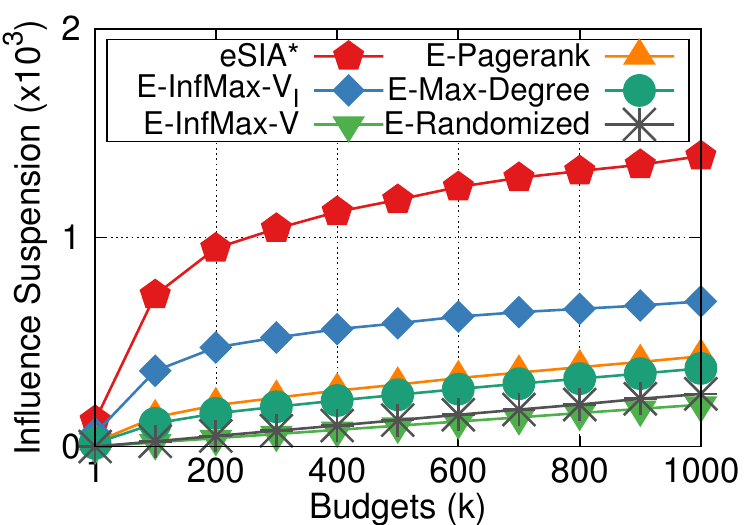}
	}
	\subfloat[Skitter]{
		\includegraphics[width=0.24\linewidth]{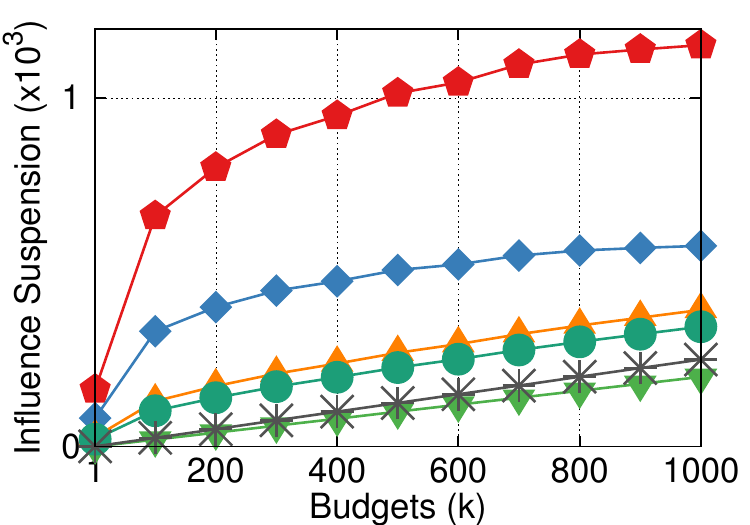}
	}
	\subfloat[LiveJournal]{
		\includegraphics[width=0.24\linewidth]{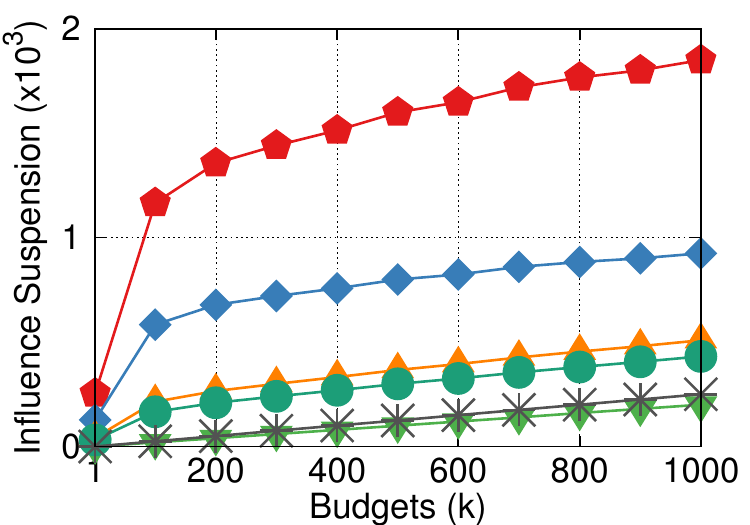}
	}
	\subfloat[Twitter]{
		\includegraphics[width=0.24\linewidth]{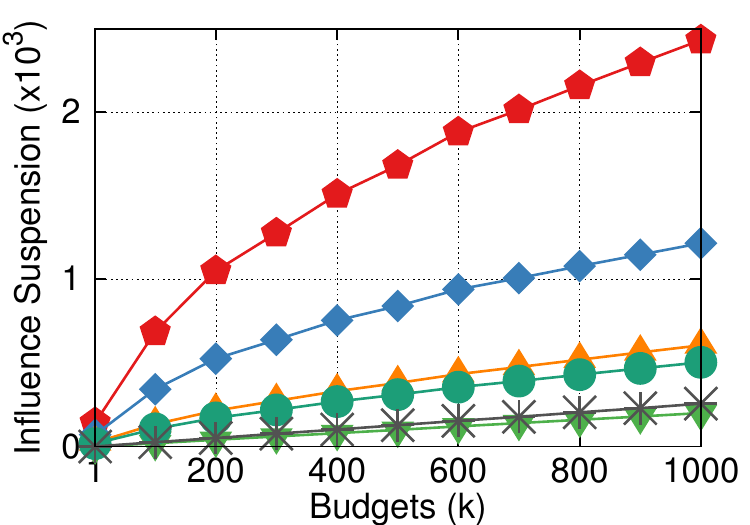}
	}
	\vspace{-0.15in}
	\caption{\small Interdiction Efficiency of different approaches on \ERI{} problem (\EIS{}* denotes the general \EIS{})}
	\label{fig:edge_inf}
	\vspace{-0.1in}
\end{figure*}

In this section, we present the results of our comprehensive experiments on  real-world networks. The results suggest the superiority of \NIS{} and \EIS{} over the other methods. 

\vspace{-0.1in}
\subsection{Experimental Settings}
\textbf{Algorithms compared}. For each of the studied problems, i.e., \NRI{} and \ERI{}, we compare three sets of algorithms:
\begin{itemize}
	\item \NIS{} and \EIS{} - our proposed algorithms, each of which has five implementations: single/multi-core CPU, and single/parallel/distributed GPU accelerations.
	\item \textsf{InfMax}-$V$ and \textsf{InfMax}-$V_I$ - algorithms for Influence Maximization problem, that finds the set of $k$ nodes in $C$ that have the highest influence.  For the edge version, we follow \cite{Khalil14} to select $k$ edges  that go \textit{into} the highest influence nodes.
	\item \textsf{GreedyCutting} \cite{Khalil14} on edge deletion problem.
	\item Baseline methods: we consider 3 common ranking measures: Pagerank, Max-Degree and Randomized.
\end{itemize}

\textbf{Datasets}. Table~\ref{tab:data_sum} provides the summary of 5 datasets used.

\setlength\tabcolsep{2pt}
\begin{table}[!htb] \small
	\caption{\small Datasets' Statistics}
	\vspace{-0.15in}
	\label{tab:data_sum}
	\centering
	\begin{tabular}{ l  r  r  c r}\toprule
		\textbf{Dataset} & \bf \#Nodes ($n$)& \bf \#Edges ($m$) & \bf Avg. Deg. & \bf Type\\
		\midrule
		\textsf{DBLP}\textsuperscript{(*)} & 655K & 2M & 6.1 & Co-author\\
		\textsf{Pokec}\textsuperscript{(*)} & 1.6M & 30.6M & 19.1 & Social\\
		\textsf{Skitter}\textsuperscript{(*)} & 1.7M & 11.1M & 6.5 & Internet\\
		\textsf{LiveJournal}\textsuperscript{(*)} & 4M & 34.7M & 8.7 & Social\\
		\textsf{Twitter}\cite{Kwak10} & 41.7M & 1.5G & 70.5 & Social\\
		\bottomrule
		\multicolumn{5}{r}{\textsuperscript{(*)}\footnotesize{http://snap.stanford.edu/data/index.html}}\\
	\end{tabular}
\end{table}


\setlength\tabcolsep{6pt}
\def\arraystretch{1}
\captionsetup{width=0.68\textwidth}
\begin{table*}[t]\small
	\begin{minipage}[t]{0.71\textwidth } \centering
		\vspace{-1.1in}
		\begin{tabular}{l | rrrrrrrrrrr}
			\toprule
			\multirow{2}{*}{\textbf{Dataset}} & \multicolumn{2}{c}{\textbf{1 CPU core}} && \multicolumn{2}{c}{\textbf{8 CPU cores}} && \multicolumn{2}{c}{\textbf{1 GPU}} &&  \multicolumn{2}{c}{\textbf{par-2 GPUs}}\\
			\cline{2-3}\cline{5-6}\cline{8-9}\cline{11-12} 
			& \textsf{time(s)} & \textsf {SpF} && \textsf{time(s)} & \textsf {SpF} && \textsf{time(s)} & \textsf {SpF} && \textsf{time(s)} & \textsf {SpF} \\
			\midrule
			\textsf{DBLP} & 10.3 & 1 && 1.7 & 6.0 && 0.1 & 103 && 0.05 & 206\\
			\textsf{Pokec} & 36.6 & 1 && 5.6 & 6.5 && 0.3 & 122 && 0.2 & 183\\
			\textsf{Skitter} & 34.1 & 1 && 5.1 & 6.5 && 0.2 & 169 && 0.1 & \textbf{338}\\
			\textsf{LiveJ} & 70.9 & 1 && 10.1 & \textbf{7.0} && 0.4 & \textbf{177} && 0.25 & 283 \\
			\textsf{Twitter} & 2517.6 & 1 && 371.2 & 6.8 && 20.5 & 123 && 12.6 & 200\\
			\bottomrule
		\end{tabular}
		\caption{\small Running time and Speedup Factor (\textsf{SpF}) \EIS{} on various platforms ($k = 100$, \textbf{par-2 GPUs} refers to parallel algorithm on 2 GPUs).}
		\label{tab:cpu_gpu}
	\end{minipage}
	\begin{minipage}[t]{0.28\textwidth}
		\includegraphics[width=0.8\textwidth]{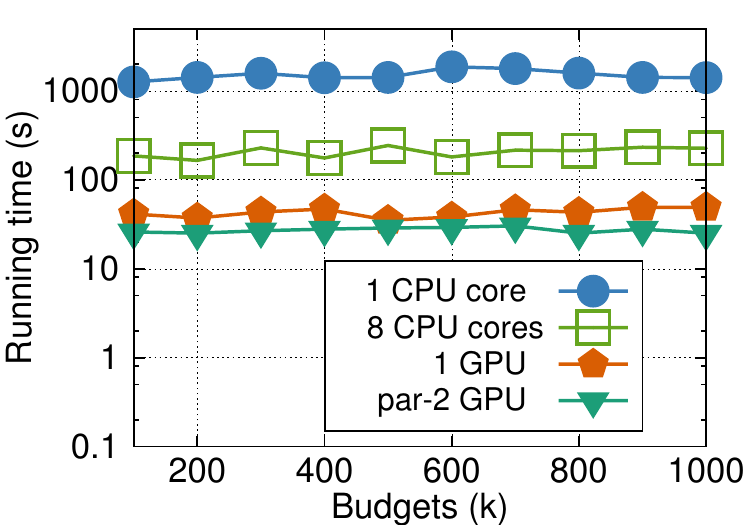}
		\vspace{-0.1in}
		\captionof{figure}{\small Running time on Twitter}
		\label{fig:cpu_gpu_twitter}
	\end{minipage}
	\vspace{-0.35in}
\end{table*}

\captionsetup{width=0.48\textwidth}

\textbf{Measurements}. We measure the performance of each algorithm in two aspects: Solution quality and Scalability. To compute the influence suspension, we adapt the EIVA algorithm in \cite{Nguyen163} to find an $(\epsilon,\delta)$-estimate $\hat \D(T, \V_I)$,
\begin{align}
	\Pr[|\hat \D(T, \V_I) - \D(T,\V_I)| \geq \epsilon\D(T,\V_I)] \leq \delta,
\end{align}
where $\epsilon,\delta$ are set to $0.01$ and $1/n$ (see details in \cite{Nguyen163}).

\textbf{Parameter Settings}. We follow a common setting in \cite{Tang15,Nguyen16,Nguyen163} and set the weight of edge $(u,v)$ to be $w(u,v) = \frac{1}{d_{in}(v)}$. Here $d_{in}(v)$ denotes the in-degree of node $v$. For simplicity, we set $C = E$ (or $C = V$) in the edge (or node) interdiction problem(s). For \NIS{} and \EIS{} algorithms, we set the precision parameters to be $\epsilon = 0.1$ and $\delta = 1/n$ as a general setting. Following the approach in \cite{Khalil14} the suspected set of nodes $\V_I$ contains randomly selected 1000 nodes with randomized probabilities between 0 and 1 of being real suspects. The budget value $k$ is ranging from 100 to 1000.

All the experiments are carried on a CentOS 7 machine having 2 Intel(R) Xeon(R) CPUs X5680 3.33GHz with 6 cores each, 2 NVIDIA's Titan X GPUs and 100 GB of RAM. The algorithms are implemented in C++ with C++11 compiler and CUDA 8.0 toolkit.

\vspace{-0.12in}
\subsection{Solution Quality}


%

\begin{figure}[!ht]
	\vspace{-0.15in}
	\subfloat[Influence Suspension]{
		\includegraphics[width=0.49\linewidth]{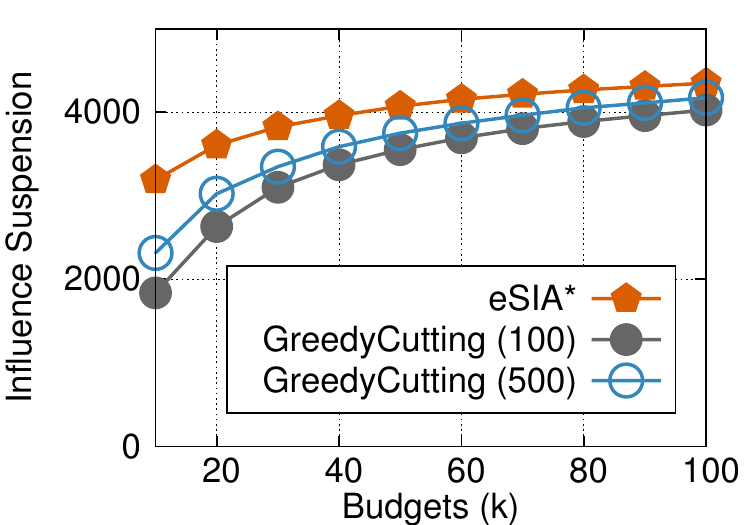}
	}
	\subfloat[Running time]{
		\includegraphics[width=0.49\linewidth]{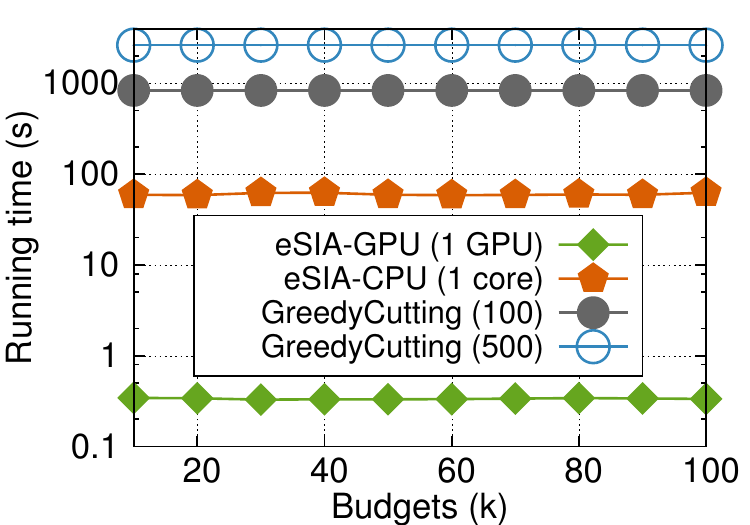}
	}
	\vspace{-0.15in}
	\caption{\small Comparison between \EIS{} and GreedyCutting with Edge Deletion Problem on Skitter network.}
	\label{fig:eis_greeycut}
\end{figure}

The results of comparing the solution quality, i.e., influence suspension, of the algorithms on four larger network datasets, e.g., Pokec, Skitter, LiveJournal and Twitter, are presented in Fig.~\ref{fig:edge_inf} for \ERI{}. Across all four datasets, we observe that \EIS{} significantly outperforms the other methods with widening margins when $k$ increases. \EIS{} performs twice as good as \textsf{InfMax-$V_I$} and many times better than the rest. Experiments on \NIS{} give similar observation and the complete results are presented in our extended version.

\textbf{Comparison with GreedyCutting \cite{Khalil14}}. We compare \EIS{} with the GreedyCutting \cite{Khalil14} which solves the slightly different Edge Deletion problem that interdicts the \textit{sum} of nodes' influences while \ERI{} minimizes the combined influence. Thus, to compare the methods for the two problems, we set the number of sources to be 1. Since we are interested in interdicting nodes with high impact on networks, we select top 10 nodes in Skitter network\footnote{Skitter is largest network that we could run GreedyCutting due to an unknown error returned by that algorithm.} with highest degrees and randomize their probabilities. We carry 10 experiments, each of which takes 1 out of 10 nodes to be the suspect. For GreedyCutting, we keep the default setting of 100 sample graphs and also test with 500 samples. We follow the edge probability settings in \cite{Khalil14} that randomizes the edge probabilities and then normalizes by,
\vspace{-0.08in}
\begin{align}
	w(u,v) = \frac{w(u,v)}{\sum_{(w,v)\in E}w(w,v)}
\end{align}
\vspace{-0.1in}

\noindent so that the sum of edge probabilities into a node is 1. Afterwards, we take the average influence suspension and running time over all 10 tests and the results are drawn in Fig.~\ref{fig:eis_greeycut}.

\textit{Results:} From Fig.~\ref{fig:eis_greeycut}, we see that clearly \EIS{} both obtains notably better solution quality, i.e., 10\% to 50\% higher, and runs substantially faster by a factor of up to 20 (CPU 1 core) and 1250 (1 GPU) than GreedyCutting. Comparing between using 100 and 500 sample graphs in GreedyCutting, we see improvements in terms of Influence Suspension when 500 samples are used, showing the \textit{quality degeneracy} of using insufficient graph samples.

\vspace{-0.1in}
\subsection{Scalability}

This set of experiments is devoted for evaluating the running time of \NIS{} and \EIS{} implementations on multi-core CPUs and single or multiple GPUs on large networks.

\vspace{-0.08in}
\subsubsection{Parallel implementations on GPU(s) vs CPUs}

We experiment the different parallel implementations: 1) single/multiple GPU, 2) multi-core CPUs to evaluate the performance in various computational platforms. Due to the strong similarity between \NIS{} and \EIS{} in terms of performance, we only measure the time and speedup factor (\textsf{SpF}) for \NIS{}. We use two Titan X GPUs for testing multiple GPUs. The results are shown in Table~\ref{tab:cpu_gpu} and Fig.~\ref{fig:cpu_gpu_twitter}.

\textbf{Running time.} From Table~\ref{tab:cpu_gpu}, we observe that increasing the number of CPUs running in parallel achieves an effective speedup of 80\% per core meaning with 8 cores,\NIS{} runs 6.5 times faster than that on a single CPU. On the other hand, using one GPU reduces the running time by 100 to 200 times while two parallel GPUs helps almost double the performance, e.g., 200 vs. 123 times faster on Twitter. Fig.~\ref{fig:cpu_gpu_twitter} confirms the speedups on different budgets.

\captionsetup{width=0.23\textwidth}
\begin{figure}[h]
	\begin{minipage}[h]{0.235\textwidth}
		\vspace{-0.2in}
		\centering
		\includegraphics[width=\linewidth]{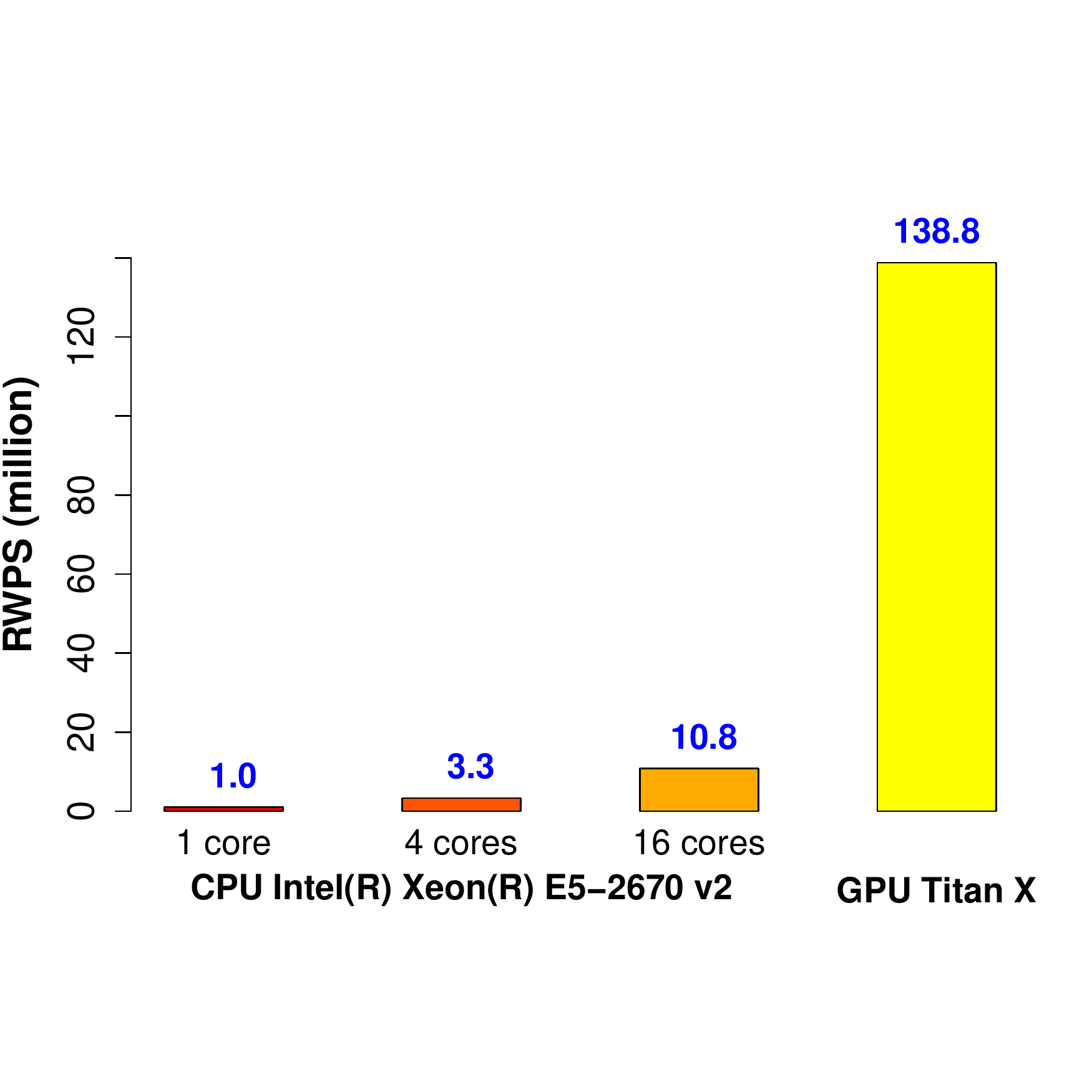}
		\vspace{-0.55in}
		\caption{\small Random walks per second on various platforms}
		\label{fig:rwps}
		\vspace{-0in}
	\end{minipage}
	\begin{minipage}[h]{0.235\textwidth }
		\vspace{-0.3in}
		\includegraphics[width=\linewidth]{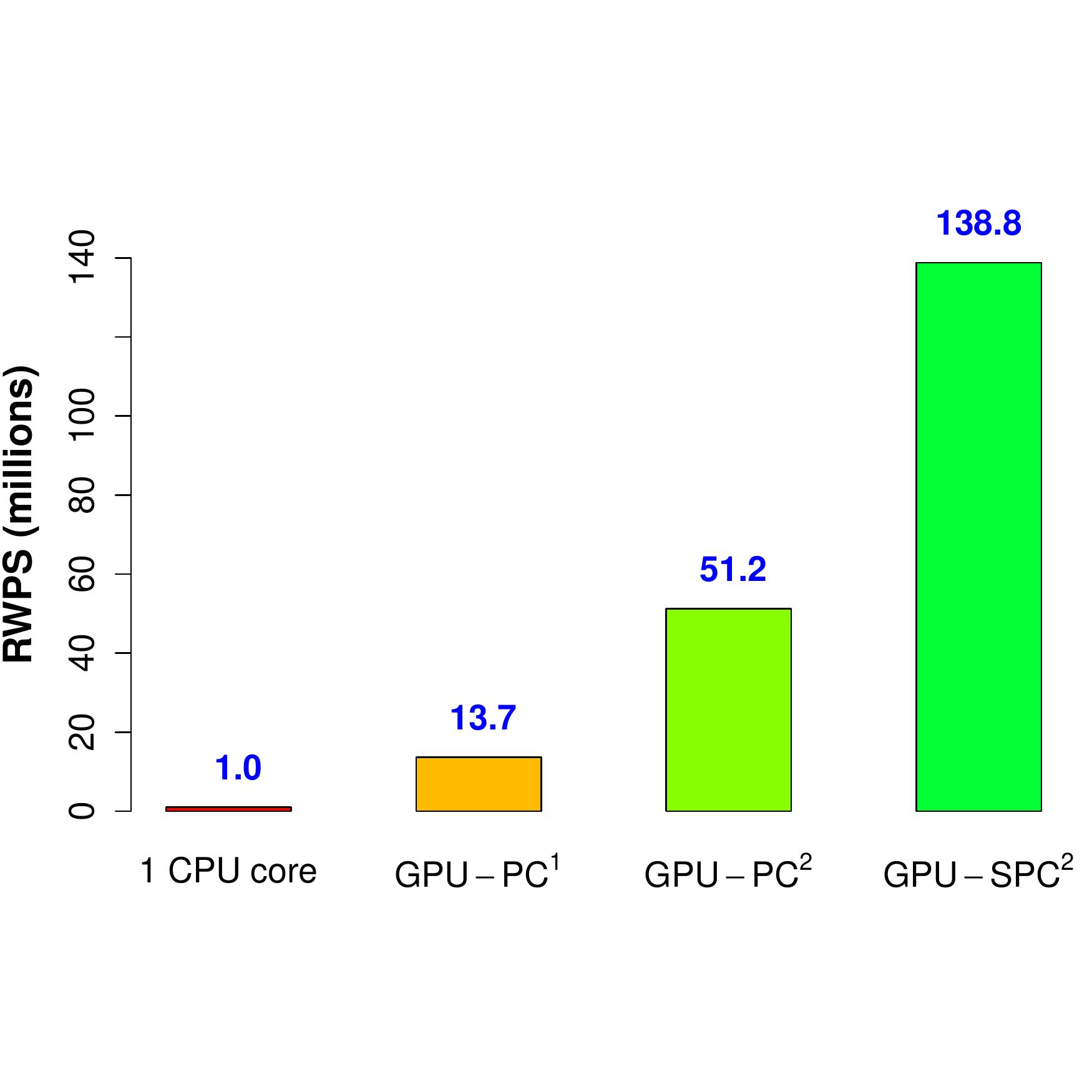}
		\vspace{-0.55in}
		\caption{\small Effects of different acceleration techniques}
		\label{fig:gpu_ver_node}
		\vspace{-0in}
	\end{minipage}
\end{figure}

\textbf{Random Walk Generating Rate.}
We compare the rates of generating random walks (samples) on different parallel platforms, i.e., GPU and CPU. The results are described in Fig.~\ref{fig:rwps}.

Unsurprisingly, the rate of random walk generation on CPU linearly depends on the number of cores achieving nearly 70\% to 80\% effectiveness. Between GPU and CPU, even with 16 cores of CPU, only 10.8 million random walks are generated per second that is around 13 times less than that on a GPU with 139 million.

\captionsetup{width=0.23\textwidth}
\begin{figure}[h]
	\begin{minipage}[h]{0.235\textwidth}
		\centering
		\includegraphics[width=\linewidth]{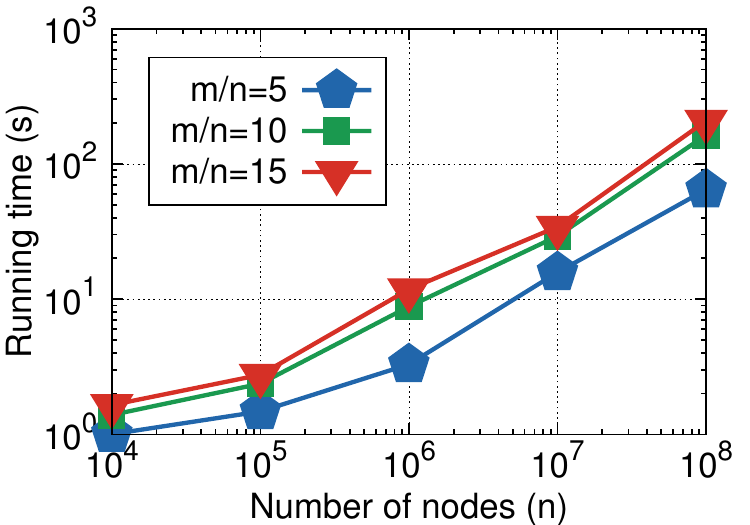}
		\vspace{-0.25in}
		\caption{Scalability tests on varying network sizes.}
		\label{fig:scale}
	\end{minipage}
	\begin{minipage}[h]{0.235\textwidth }
		\centering
		\includegraphics[width=0.9\linewidth]{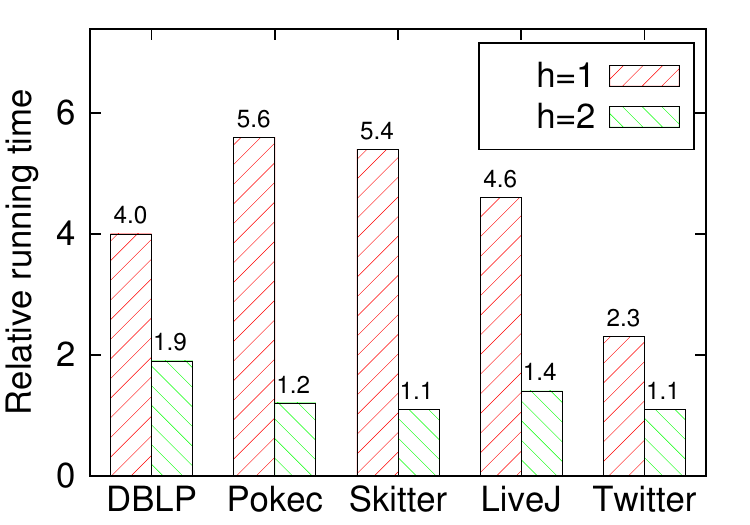}
		\vspace{-0.0in}
		\caption{\small Distributed algorithm on 2 GPUs.}
		\label{fig:dis}
	\end{minipage}
\end{figure}
\textbf{Scalability Test.}
We carry another test on the scalability of our GPU implementation. We create synthetic networks by GTgraph \cite{Gtgraph} with number of nodes, denoted by $n$, increasing from tens of thousands, $10^4$, to hundreds of millions, $10^8$. For each value of $n$, we also test on multiple densities, i.e., ratio of edges to nodes $m/n$. Specifically, we test with densities of $5,10$ and $15$. Our results are plotted in Figure~\ref{fig:scale}. The results show that the running time of \NIS{} increases almost linearly with the size of the network.

\vspace{-0.1in}
\subsubsection{Distributed algorithm on multiple GPUs}
We implemented our distributed \NIS{} algorithm on two GPUs and compared the performance with that on a single GPU. For the distributed version, we test on two values of extension parameter $h = 1$ and $h = 2$. The results are presented in Fig.~\ref{fig:dis}. We see that the distributed algorithm on multiple GPUs is several times slower than on a single GPU. However, this can be addressed by extending each partition to include nodes which are at most two hops away.
\vspace{-0.1in}
\subsubsection{Effects of Acceleration Techniques on GPUs}
We experimentally evaluate the benefit of our acceleration techniques. We compare 3 different versions of \NIS{} and \EIS{}: 1) \textsf{GPU-PC$^1$} which employs $O(1)$-space Path-encoding and $O(1)$-space Cycle-detection by the slow Floyd's algorithm; 2) \textsf{GPU-PC$^2$} which employs $O(1)$-space Path-encoding and $O(1)$-space Cycle-detection by the fast Brent's algorithm; 3) \textsf{GPU-SPC$^2$} which applies all the techniques including the empirical Sliding-window early termination. We run four versions on all the datasets and compute the RWPS compared to that on a single-core CPU. The average results are in Fig.~\ref{fig:gpu_ver_node}.

The experiment results illustrate the huge effectiveness of the acceleration techniques. Specifically, the $O(1)$-space Path-encoding combined with the slow Floyd's algorithm for Cycle-detection (\textsf{GPU-PC$^1$}) helps boost up the performance by 14x. When the fast Brent's algorithm is incorporated for $O(1)$-space Cycle-detection, the speedup is further increased to 51x while applying all the techniques effectively improves the running time up to 139x faster.
	\vspace{-0.18in}
\section{Related Work}
\label{sec:relatedwork}
Severa works have been proposed for removing/adding nodes/edges to minimize or maximize the influence of a node set in a network. \cite{Kimura08,Kuhlman13} proposed heuristic algorithms under the linear threshold model and its deterministic version. \cite{He12} studies the influence blocking problem under the competitive linear threshold model, that selects $k$ nodes to initiate the inverse influence propagation to block the initial cascade. Misinformation containment has also been widely studied in the literature \cite{Nguyen13,Kumar14}. Other than LT model, the node and edge interdiction problems were studied under other diffusion models: \cite{Tong12} consider the SIR (Susceptible-Infected-Recovery) model while \cite{Kimura09} considers the IC model.

The closest to our work is  \cite{Khalil14} in which the authors study two problems under the LT model: removing $k$ edges to minimize the sum over influences of nodes in a set and adding $k$ edges to maximize the sum. They prove the monotonicity and submodularity of their objective functions and then develop two approximation algorithms for the two corresponding problems. However, their algorithms do not provide a rigorous approximation factor due to relying on a fixed number of simulations. In addition, there is no efficient implementation for billion-scale networks.

Another closely related line of works is on Influence Maximization \cite{Kempe03,Du13,Tang15,Nguyen163} which selects a set of seed node that maximizes the spread of influence over the networks. Chen et al. \cite{Chen100} proved that estimating influence of a set of nodes is \#P-hard by counting simple paths (self-avoiding walks). Learning the parameters in propagation model have equally attracted great research interest \cite{Goyal10,Cha10}. Network interdiction problems have intensively studied, e.g., interdicting maximum flow, shortest path, minimum spanning tree and many others (see \cite{Zenklusen15} and the references therein).

GPUs have recently found effective uses in parallelizing and accelerating the practical performance of many problems. Related to our spread interdiction problems, Liu et al. \cite{Liu14} propose an GPU-accelerated Greedy algorithm for the well-studied Influence Maximization problem to process sample graphs. 
In another direction, \cite{Merrill12,Liu16} and a several follow-ups study GPUs on the fundamental Breadth-First-Search problem and achieve good performance improvement.

	\vspace{-0.15in}	
	\section{Conclusion}
    \label{sec:con}
	This paper aims at stopping an epidemic in a stochastic networks $\G$ following the popular Linear Threshold model. The problems ask for a set of nodes (or edges) to remove from $\G$ such that the influence after removal is minimum or the \textit{influence suspension} is maximum. We draw an interesting connection between the Spread Interdiction problems and the concept of \textit{Self-avoiding Walk} (\saw{}). We then propose two near-optimal approximation algorithms. To accelerate the computation, we propose three acceleration techniques for parallel and distributed algorithms on GPUs. Our algorithms show large performance advantage in both solution quality and running time over the state-of-the-art methods.
	
	%
	\vspace{-0.15in}
	\bibliographystyle{ACM-Reference-Format}
	\bibliography{dijkstra,infection,isi,nphard,pids,social,targetedIM,randgen,gpu,networkInterdiction,saw}
	%
	%
    \balance
	\newpage
\appendix

\section{Proofs of Lemmas and Theorems}
We summarize the commonly used notations in Table~\ref{tab:syms}.
\renewcommand{\arraystretch}{1.3}

\setlength\tabcolsep{3pt}
\begin{table}[h]\small
	\centering
	\caption{Table of notations}
	\vspace{-0.1in}
	\begin{tabular}{p{1.5cm}|p{6.5cm}}
		\addlinespace
		\toprule
		\bf Notation  &  \quad \quad \quad \bf Description \\
		\midrule 
		$n, m$ & \#nodes, \#edges of graph $\G=(V, E, w)$.\\
		\hline
		$G \sim \G$ & A sample graph $G$ of $\G$\\
		\hline
		$\I_\G(\V_I)$ & Influence Spread of $\V_I$ in $\G$.\\
		\hline
		$\OPTe$ & The maximum influence suspension by removing at most $k$ edges.\\
		\hline
		$\OPTn$ & The maximum influence suspension by removing at most $k$ nodes.\\
		\hline
		$\hat T_k, \hat S_k$ & The returned size-$k$ edge set of \EIS{} and \NIS{}.\\
		\hline
		$T^*_k, S^*_k$ & An optimal size-$k$ set of edges and nodes.\\
		\hline
		$\R_t, \R'_t$ & Sets of random \hsaw{} samples in iteration $t$.\\
		\hline
		$\Cov_{\R_t}(T)$ & \#\hsaw{} $h_j \in \R_t$ intersecting with $T$.\\           
		\hline
		$\Lambda$ &
		$\Lambda = (2+\frac{2}{3}\epsilon)\ln (\frac{3 t_{\max}}{\delta}) \frac{1}{\epsilon^2},$.\\
		\hline
		$\Lambda_1$ & $\Lambda_1 = 1+(1+\epsilon)(2+\frac{2}{3}\epsilon)\ln (\frac{3 t_{\max}}{\delta}) \frac{1}{\epsilon^2},$\\
		\bottomrule
	\end{tabular}
	\label{tab:syms}
\end{table}

\subsection*{Proof of Theorem~\ref{theo:hard}}
We prove that both \NRI{} and \ERI{} cannot be approximated within $1-1/e-o(1)$, that also infers the NP-hardness of these two problems.

\textbf{\NRI{} cannot be approximated within $1-1/e-o(1)$.}
We prove this by showing that Influence Maximization problem \cite{Kempe03} is a special case of \NRI{} with a specific parameter setting. Considering an instance of Influence Maximization problem which finds a set of $k$ nodes having the maximum influence on the network on a probabilistic graph $\G = (V,E,w)$, we construct an instance of \NRI{} as follows: using the same graph $\G$ with $\V_I = V$ and $\forall v \in V, p(v) = 1/2$ and candidate nodes are all the nodes in the graph, i.e., $C = V$.

On the new instance, for a set $S$ of $k$ nodes, we have,
\begin{align}
	\D(S, \V_I) &= \I_{\G}(\V_I) - \I_{\G'}(\V_I) \nonumber \\
	&= \sum_{X \sim \V_I} (\I_{\G}(X) - \I_{\G'}(X)) \Pr[X \sim \V_I].
\end{align}
Since $\forall v \in V, p(v) = 1/2$, from Eq.~\ref{eq:prob_src_set}, we have,
\begin{align}
	\Pr[X \sim \V_I] = 1/2^n.
\end{align}
Thus,
\begin{align}
	\D(S, \V_I) &= \frac{1}{2^n}\sum_{X \sim \V_I} (\I_{\G}(X) - \I_{\G'}(X)) \nonumber \\
	&= \frac{1}{2^n}\sum_{X \sim \V_I} \sum_{v \in V} (\I_{\G}(X,v) - \I_{\G'}(X,v)),
	\label{eq:dsvi}
\end{align}
where $G, G'$ are sampled graph from $\G$ and $\G'$. We say $G$ is \textit{consistent}, denoted by $G \propto G'$, if every edge $(u,v)$ appearing in $G'$ is also realized in $G$. Thus, each sampled graph $G'$ of $\G'$ corresponds to a class of samples $G$ of $\G$. We define it as the consistency class of $G'$ in $\G$, denoted by $\mathcal{C}_{G'} = \{ G\sim \G | G \propto G' \}$. More importantly, we have,
\begin{align}
	\Pr[G' \sim G] = \sum_{G \in \mathcal{C}_{G'}} \Pr[G \sim \G].
\end{align}
Note that if $G'_1 \neq G'_2$ are sampled graph from $\G'$, then $\mathcal{C}_{G'_1} \cap \mathcal{C}_{G'_2} = \emptyset$. Thus, we obtain,
\begin{align}
	\I_{\G'}(X,v)& = \sum_{G' \sim \G'}\chi^{G'}(X,v)\Pr[G' \sim \G'] \nonumber \\
	& = \sum_{G \sim \G} (\chi^G(X,v)- \chi^G(X,S,v) ) \Pr[G \sim \G],
\end{align}
where 
\begin{align}
\chi^G(X,S,v) = \twopartdef {1} {v \text{ is only reachable from } X \text{ through } S} {0} \nonumber
\end{align}
Hence, $\I_{\G'}(X,v) = \I_{\G}(X,v) - \sum_{G \sim \G} \chi^G(X,S,v) \Pr[G \sim \G]$. Put this to Eq.~\ref{eq:dsvi}, we have,
\begin{align}
	\D(S, \V_I) &= \frac{1}{2^n} \sum_{G \sim \G} \sum_{v \in V} \sum_{X \sim \V_I} \chi^G(X,S,v) \Pr[G \sim \G] \nonumber \\
	&= \frac{1}{2^n} \sum_{G \sim \G} \sum_{v \in \I_G(S)} \sum_{X \sim \V_I} \chi^G(X,S,v) \Pr[G \sim \G] \nonumber\\
	& = \frac{1}{2^n} \sum_{G \sim \G} \sum_{v \in \I_G(S)} 2^{n-1} \Pr[G \sim \G] = \frac{1}{2} \I_{\G}(S),
\end{align}
where the third equality is due to the property of the LT model that for a node, there exists at most one incoming edge in any sample $G \sim \G$.

Therefore, we have $\D(S, \V_I) = 1/2 \I_{\G}(S)$ where $\I_{\G}(S)$ is the influence function which is well-known to be NP-hard and not able to be approximated within $1-1/e-o(1)$. Thus, $\D(S, \V_I)$ possesses the same properties.

\textbf{\ERI{} cannot be approximated within $1-1/e-o(1)$. } Based on the hardness of \NRI{}, we can easily prove that of the \ERI{} by a reduction as follow: assuming an instance of \ERI{} on $\G = (V,E,w)$ and $\V_I$, for $(u,v) \in E$, we add a node $e_{uv}$ and set $w_{ue_{uv}} = w_{uv}, w_{e_{uv}v} = 1$. We also restrict our node selection to $e_{uv}$ where $(u,v) \in E$. As such, the \ERI{} is converted to a restricted \NRI{} problem which is NP-hard and cannot be approximated within $1-1/e-o(1)$.

\subsection*{Proof of Theorem~\ref{lem:rr_edge}}

From the definition of influence suspension function for a set $T_k$ of edges (Eq.~\ref{eq:sus_nod}), we have,
\begin{align}
\DE(T_k, &\V_I) = \I_{\G}(\V_I) - \I_{\G'}(\V_I)\nonumber \\
&= \sum_{X \sim \V_I} [\I_{\G}(X) - \I_{\G'}(X)] \Pr[X \sim \V_I] \nonumber \\
&= \sum_{X \sim \V_I} \sum_{v \in V} [\I_{\G}(X,v) - \I_{\G'}(X,v)] \Pr[X \sim \V_I].
\label{eq:decompose}
\end{align}

In Eq.~\ref{eq:decompose}, the set $X$ is a sample set of $\V_I$ and deterministic. We will extend the term $\I_{\G}(X,v) - \I_{\G'}(X,v)$ inside the double summation and then plug in back the extended result. First, we define the notion of collection of \hsaw{}s from $X$ to a node $v$.

\textit{Collection of \hsaw{}s}. In the original stochastic graph $\G$ having a set of source nodes $X$, for a node $v$, we define a collection $\mathcal{P}_{X,v}$ of \hsaw{}s to include all possible \hsaw{}s $h$ from a node in $X$ to $v$,
\begin{align}
\mathcal{P}_{X,v} = \{h = <v_1=v,v_2,\dots,v_l> | h \cap X = \{v_l\}\}.
\end{align} 


According to Eq.~\ref{eq:compute_inf}, the influence of a seed set $X$ onto a node $v$ has an equivalent computation based on the sample graphs as follows,
\begin{align}
\label{eq:inf_x_v}
\I_{\G}(X,v) = \sum_{G \sim \G}\chi^G(X,v) \Pr[G \sim \G],
\end{align}
where $\chi^G(X,v)$ is an indicator function having value 1 if $v$ is reachable from $X$ by a live-edge path in $G$ and 0 otherwise. If we group up the sample graphs according to the \hsaw{} from nodes in $X$ to $v$ such that $\Omega_{h}$ contains all the sample graphs having the path $h$. Then since the set $X$ is deterministic,
\begin{align}
\label{eq:walk_prob}
{\Pr}_{X\sim \V_I}[h] = \sum_{G\sim \Omega_{h}} \Pr[G \sim \G].
\end{align}

Due to the walk uniqueness property, $\Omega_{h}$ for $h \in \mathcal{P}_{X,v}$ are completely disjoint and their union is equal the set of sample graphs of $\G$ that $v$ is activated from nodes in $X$. Thus Eq.~\ref{eq:inf_x_v} is rewritten as,
\begin{align}
\I_{\G}(X,v) = \sum_{h \in \mathcal{P}_{X,v}} \sum_{G \in \Omega_{h}} \Pr[G \sim \G] = \sum_{h \in \mathcal{P}_{X,v}} {\Pr}_{X \sim \V_I}[h]. \nonumber
\end{align}

We now compute the value of $\I_{\G}(X,v) - \I_{\G'}(X,v)$ in the summation of Eq.~\ref{eq:decompose}. Since $\G'$ is induced from $\G$ by removing the edges in $T_k$, the set of all possible sample graphs of $\G'$ will be a subset of those sample graphs of $\G$. Furthermore, if $G \sim \G$ and $G$ can not be sampled from $\G'$, then $\Pr[G \sim \G'] = 0$ and we have the following,
\begin{align}
\I_{\G}&(X,v) - \I_{\G'}(X,v) = \sum_{h \in \mathcal{P}_{X,v}} \sum_{G \in \Omega_{h}}(\Pr[G \sim \G] - \Pr[G \sim \G']) \nonumber \\
&= \sum_{\substack{h \in \mathcal{P}_{X,v}\\
		h \cap T_k = \emptyset}} \sum_{G \in \Omega_{h}}(\Pr[G \sim \G] - \Pr[G \sim \G'])\nonumber \\
& \qquad \qquad + \sum_{\substack{h \in \mathcal{P}_{X,v}\\
		h \cap T_k \neq \emptyset}} \sum_{G \in \Omega_{h}}(\Pr[G \sim \G] - \Pr[G \sim \G'])\nonumber \\
& = \sum_{\substack{h \in \mathcal{P}_{X,v}\\
		h \cap T_k = \emptyset}} ({\Pr}_{X\sim \V_I}[h] - {\Pr}_{X\sim \V_I}[h]) + \sum_{\substack{h \in \mathcal{P}_{X,v}\\
		h \cap T_k \neq \emptyset}} {\Pr}_{X\sim \V_I}[h] \nonumber \\
& = \sum_{\substack{h \in \mathcal{P}_{X,v}\\
		h \cap T_k \neq \emptyset}} {\Pr}_{X\sim \V_I}[h],
\end{align}
where the second equality is due to the division of $\mathcal{P}_{X,v}$ into two sub-collections of \hsaw{}s. The third and forth equalities are due to Eq.~\ref{eq:walk_prob} when $X$ is deterministic. Thus, we obtain,
\begin{align}
	\label{eq:eq14}
	\I_{\G}(X,v) &- \I_{\G'}(X,v) = \sum_{\substack{h \in \mathcal{P}_{X,v}\\
		h \cap T_k \neq \emptyset}} {\Pr}_{X\sim \V_I}[h],
\end{align}
which is the summation over the probabilities of having a \hsaw{} appearing in $\G$ but not in $\G'$ indicated by $h$ is suspended by $T_k$ given that $X$ is deterministic.

Plugging Eq.~\ref{eq:eq14} back to Eq.~\ref{eq:decompose}, we obtain,
\begin{align}
	\DE(T_k,& \V_I) = \sum_{X \sim \V_I} \sum_{v \in V} [\I_{\G}(X,v) - \I_{\G'}(X,v)] \Pr[X \sim \V_I] \nonumber \\
	&= \sum_{X \sim \V_I} \sum_{v \in V} \sum_{\substack{h \in \mathcal{P}_{X,v}\\
		h \cap T_k \neq \emptyset}} {\Pr}_{X\sim \V_I}[h] \Pr[X \sim \V_I].
\label{eq:eq19}
\end{align}

We define $\mathcal{P}_X$ to be the set of all \hsaw{}s from a node in $X$ to other nodes and $\mathcal{P}$ to be the set of all \hsaw{}s from nodes in the probabilistic set $\V_I$ to other nodes. Then Eq.~\ref{eq:eq19} is rewritten as,
\begin{align}
	& \DE(T_k, \V_I) = \sum_{X \sim \V_I} \sum_{\substack{h \in \mathcal{P}_{X}\\
		h \cap T_k \neq \emptyset}} {\Pr}_{X\sim \V_I}[h] \Pr[X \sim \V_I] \nonumber \\
		\label{eq:final_eq}
	& = \sum_{\substack{h \in \mathcal{P}\\
		h \cap T_k \neq \emptyset}} \Pr[h] = \sum_{
		h \cap T_k \neq \emptyset} \Pr[h] \\
	& = \I_{\G}(\V_I) \frac{ \sum_{
				h \cap T_k \neq \emptyset}\Pr[h]}{\I_{\G}(\V_I)} = \I_{\G}(\V_I) \frac{ \sum_{
				h \cap T_k \neq \emptyset}\Pr[h]}{\sum_{
				h \in \mathcal{P}}\Pr[h]} \\
	& = \I_{\G}(\V_I) \Pr[T \text{ interdicts }h]
\end{align}
The last equality is obtained from $\Pr[T \text{ interdicts }h] = \frac{ \sum_{
		h \cap T_k \neq \emptyset}\Pr[h]}{\sum_{
		h \in \mathcal{P}}\Pr[h]}$ which holds since $h$ is random \hsaw{}.

\textbf{Proof of Monotonicity and Submodularity of $\DE(T_k, \V_I)$.} The left-hand side of Eq.~\ref{eq:final_eq} is equivalent to the \textit{weighted coverage} function of a set cover system in which: every \hsaw{} $h \in \mathcal{P}$ is an element in a universal set $\mathcal{P}$ and edges in $E$ are subsets. The subset of $e \in E$ contains the elements that the corresponding \hsaw{}s have $e$ on their paths. The probability $\Pr[h]$ is the weight of element $h$. Since the weighted coverage function is monotone and submodular, it is followed that $\DE(T_k, \V_I)$ has the same properties.

\section{Proof of Theorem~\ref{theo:app}}
Before proving Theorem~\ref{theo:app}, we need the following results.

Let $R_1,\ldots, R_N$ be the random \hsaw{} samples generated in \EIS{} algorithms. Given a subset of edges $T \subset E$, define $X_j(T) = \min\{|R_j \cap T|,1\}$, the Bernouli random variable with mean $\mu_X = \E[X_j(T)]=\DE(T)/\I(\V_I, p)$. Let $\hat \mu_X = \frac{1}{N}\sum_{i=1}^{N}X_i(T)$ be an estimate of $\mu_X$. Corollaries 1 and 2 in \cite{Tang15} state that,
\begin{Lemma}[\cite{Tang15}]
	\label{lem:chernoff}
	For $N > 0$ and $\epsilon > 0$, it holds that,
	\vspace{-0.05in}
	\begin{align}
	\label{eq:plus}
	\Pr[\hat \mu_X >  (1+ \epsilon) \mu_X]  &\leq
	\exp{(\frac{-N\mu_X\epsilon^2}{2 + \frac{2}{3}\epsilon})},\\
	\label{eq:minus}
	\Pr[\hat \mu_X <  (1- \epsilon) \mu_X] &\leq \exp{(\frac{-N\mu_X\epsilon^2}{2})}.
	\end{align}
\end{Lemma}
The above lemma is used in proving the estimation guarantees of the candidate solution $\hat T_k$ found by \textsf{Greedy} algorithm in each iteration and the optimal solution $T^*_k$.

Recall that \EIS{} stops when either 1) the number of samples exceeds the cap, i.e., $|\R_t| \geq N_{\max}$ or 2) $\epsilon_t \leq \epsilon$ for some $t\geq 1$. In the first case, $N_{\max}$ was chosen to guarantee that 
$\hat T_k$ will be a $(1-1/e-\epsilon)$-approximation solution w.h.p.
\begin{Lemma}
	\label{lem:cap}
	Let $B^{(1)}$ be the bad event that 
	\[
	B^{(1)} = (|\R_t| \geq N_{\max}) \cap (\DE(\hat T_k) < (1-1/e-\epsilon)\emph{\OPT}^{(e)}_k).
	\] We have
	\[
	\Pr[B^{(1)}] \leq  \delta/3.
	\]
\end{Lemma}
\begin{proof}
	We prove the lemma in two steps: 
	\begin{itemize}
		\item[(S1)] With $N = (2-\frac{1}{e})^2(2+\frac{2}{3}\epsilon) \I_\G(\V_I) \cdot \frac{\ln (6/\delta)+\ln {m \choose k}}{\OPTe \epsilon^2}$ \hsaw{} samples, the returned solution $\hat T_k$ is an $(1-1/e-\epsilon)$-approximate solution with probability at least $1-\delta/3$.
		\item[(S2)] $N_{\max} \geq N$.
	\end{itemize}

	\textit{Proof of (S1).} Assume an optimal solution $T^*_k$ with maximum influence suspension of $\OPTe{}$. Use $N = (2-\frac{1}{e})^2(2+\frac{2}{3}\epsilon) \I_\G(\V_I) \cdot \frac{\ln (6/\delta)+\ln {m \choose k}}{\OPTe \epsilon^2}$ \hsaw{} samples and apply Lemma~\ref{lem:chernoff} on a set $T_k$ of $k$ edges, we obtain,
	\begin{align}
		& \Pr[\DE_{t}(T_k) \geq \DE(T_k) + \frac{\epsilon}{2-1/e} \OPTe] \\
		& = \Pr[\DE_{t}(T_k) \leq \left (1 + \frac{\epsilon}{2-1/e} \frac{\OPTe}{\DE(T_k)} \right) \DE(T_k)]\\
		& \leq \exp \left( - \frac{N \DE(T_k) }{(2+2/3\epsilon)\I_\G(\V_I)} \left ( \frac{\OPTe (2-1/e)}{\DE(T_k)\epsilon}\right )^2 \right) \\
		& \leq \frac{\delta {m \choose k}}{6}
	\end{align}
	
	Applying union bound over all possible edge sets of size $k$ and since $\hat T_k$ is one of those sets, we have,
	\begin{align}
		\label{eq:eq43}
		\Pr[\DE_{t}(T_k) \geq \DE(T_k) + \frac{\epsilon}{2-1/e} \OPTe] \leq \frac{\delta}{6}
	\end{align}
	
	Similarly, using the same derivation on the optimal solution $T^*_k$ and apply the second inequality in Lemma~\ref{lem:chernoff}, we obtain,
	\begin{align}
		\label{eq:eq44}
		\Pr[\DE_{t}(T^*_k) \leq (1 - \frac{\epsilon}{2-1/e}) \OPTe] \leq \frac{\delta {m \choose k}}{6} 
	\end{align}
	
	Eqs.~\ref{eq:eq43} and~\ref{eq:eq44} give us the bounds on two bad events:
	\begin{itemize}
		\item[(1)] $\DE_{t}(T_k) \geq \DE(T_k) + \frac{\epsilon}{2-1/e} \OPTe$ and,
		\item[(2)] $\DE_{t}(T^*_k) \leq (1 - \frac{\epsilon}{2-1/e}) \OPTe$
	\end{itemize}
	with the maximum probability on either of them happening is $\frac{\delta}{6} + \frac{\delta {m \choose k}}{6} \leq \frac{\delta}{3}$. Thus, in case neither of the two bad events happens, we have both,
	\begin{itemize}
		\item[(1')] $\DE_{t}(T_k) \leq \DE(T_k) + \frac{\epsilon}{2-1/e} \OPTe$ and,
		\item[(2')] $\DE_{t}(T^*_k) \geq (1 - \frac{\epsilon}{2-1/e}) \OPTe$
	\end{itemize}
	with probability at least $1-\frac{\delta}{3}$. Using (1') and (2'), we derive the approximation guarantee of $\hat T_k$ as follows,
	\begin{align}
		\DE_{t}(T_k) & \leq \DE(T_k) + \frac{\epsilon}{2-1/e} \OPTe \nonumber \\
		\Leftrightarrow \text{ } \DE(T_k) & \geq \DE_t(T_k) - \frac{\epsilon}{2-1/e} \OPTe \nonumber \\
		& \geq (1-1/e)\DE_t(T^*_k) - \frac{\epsilon}{2-1/e} \OPTe \nonumber \\
		& \geq (1-1/e) (1-\frac{\epsilon}{2-1/e})\OPTe - \frac{\epsilon}{2-1/e} \OPTe \nonumber \\
		& \geq (1-1/e - ((1-1/e)\frac{\epsilon}{2-1/e} + \frac{\epsilon}{2-1/e}))\OPTe \nonumber \\
		& \geq (1-1/e-\epsilon)\OPTe
	\end{align}
	Thus, we achieve $\DE(T_k) \geq (1-1/e-\epsilon)\OPTe$ with probability at least $1-\frac{\delta}{3}$
	
	\textit{Proof of (S2).} It is sufficient to prove that $\frac{k}{m} \leq \frac{\OPTe}{\I(\V_I, p)}$ which is trivial since it equivalent to $\OPTe \geq \frac{k}{m} \I(\V_I, p)$ and the optimal solution $\OPTe$ with $k$ edges must cover at least a fraction $\frac{k}{m}$ the total influence of $\I_\G(\V_I)$. Note that there are $m$ edges to select from and the influence suspension of all $m$ edges is exactly $\I_\G(\V_I)$.
\end{proof}

In the second case, the algorithm stops when $\epsilon_t \leq \epsilon$ for some
$1\leq t \leq t_{\max}$. The maximum number of iterations $t_{\max}$ is bounded by $O(\log_2 n)$ as stated below.
\begin{Lemma}
	\label{lem:tmax}
	The number of iterations in \emph{\EIS{}} is at most $t_{\max} = O(\log n)$.
\end{Lemma}
\begin{proof}
	Since the number of \hsaw{} samples doubles in every iteration and we start at $\Lambda$ and stop with at most $2 N_{\max}$ samples, the maximum number of iterations is,
	\begin{align}
	t_{\max} & = \log_2 (\frac{2 N_{\max}}{\Upsilon(\epsilon, \delta/3)}) = \log_2 \left (2(2-\frac{1}{e})^2\frac{(2+\frac{2}{3}\epsilon)m \cdot \frac{\ln (6/\delta)+\ln {m\choose k}}{k\epsilon^2}}{(2+\frac{2}{3}\epsilon)\ln (\frac{3}{\delta}) \frac{1}{\epsilon^2}} \right) \nonumber \\
	& = \log_2 \left(2(2-\frac{1}{e})^2 \frac{m (\ln (6/\delta)+ \ln{m \choose k})}{k \ln (3/\delta)} \right) \nonumber \\
	& \leq \log_2 \left(2 (2-\frac{1}{e})^2\frac{m (\ln (6/\delta)+ k \ln m)}{k \ln (3/\delta)} \right) \nonumber \\
	& \leq \log_2 \left(2 (2-\frac{1}{e})^2 \frac{m}{k} + 2(2-\frac{1}{e})^2 m \frac{\ln m}{\ln (3/\delta)} + 2(2-\frac{1}{e})^2 \frac{m \ln 2}{k \ln (3/\delta)} \right) \nonumber \\
	& = O(\log_2 n)
	\end{align}
	The last equality is due to that $k \leq m \leq n^2$; $\textsf{EPT}_k$ is constant and our precision parameter $1/\delta = \Omega(n)$.
\end{proof}

For each iteration $t$, we will bound the probabilities of the bad events that lead to inaccurate estimations of  $\DE(\hat T_k)$ through $\R'_t$, and $\DE(T^*_k)$ through $\R_t$(Line~5 in Alg.~\ref{alg:check}). We obtain the following.

\begin{Lemma}
	\label{lem:bad2}
	For each $1\leq t \leq t_{\max}$, let 
	\[
	\hat \epsilon_t \text{ be the unique root of } f(x)=\frac{\delta}{3t_{\max}},
	\]
	where $f(x)=\exp{\left(-\frac{N_t \frac{\DE(\hat T_k)}{\I_\G(\V_I)} x^2 }{2+2/3x}\right)}$, and 
	\[
	\epsilon_t^*=
	\epsilon \sqrt{\frac{\I_\G(\V_I)}{(1+\epsilon/3)2^{t-1} \emph{\textsf{OPT}}^{(e)}_k}}.
	\]
	Consider the following bad events
	\begin{align*}
		B_t^{(2)} &= \left(
		\DE_{t'}(\hat T_k) \geq (1+\hat \epsilon_t)\DE(\hat T_k) \right),
		\\
		B_t^{(3)} &= \left(
		\DE_{t}(T^*_k) \leq (1-\epsilon_t^*) \emph{\textsf{OPT}}^{(e)}_k
		\right).
	\end{align*}
	We have 
	\[
	\Pr[B_t^{(2)}],  \Pr[B_t^{(3)}] \leq \frac{\delta}{3t_{\max}}.
	\]
\end{Lemma}
\begin{proof}
	One can verify that $f(x)$ is a strictly decreasing function for $x>0$. Moreover, $f(0)=1$ and $\lim_{x\rightarrow \infty} f(x) = 0$. Thus, the equation $f(x)=\frac{\delta}{3t_{\max}}$ has an {\em unique solution} for $0<\delta<1$ and $t_{\max}\geq 1$.
	
	{\em Bound the probability of $B^{(2)}_t$}: 
	Note that $\hat \epsilon_t$ and the samples generated in $\R'_t$ are independent. Thus, we can apply the concentration inequality in Eq.~(\ref{eq:plus}):
	\begin{align*}
	\Pr[\DE_{t'}(\hat T_k) \geq (1 + \hat \epsilon_t)\DE(\hat T_k)]
	& \leq \exp\left({-\frac{N_t \DE(\hat T_k) {\hat \epsilon_t}^2}{(2+\frac{2}{3}\hat \epsilon_t) \I_\G(\V_I)}}\right) \\ 
	& \leq \frac{\delta}{3 t_{\max}}.
	\end{align*}
	The last equation is due to the definition of $\hat \epsilon_t$.	
	
	{\em Bound the probability of $B^{(3)}_t$}: Since $\epsilon^*_t$ is fixed and independent from the generated samples, we have
	\begin{align}
	\Pr[& \DE_t(T^*_k) \leq (1-\epsilon^*_t) \OPTe]  \leq \exp\left({- \frac{|\R_t| \OPTe {\epsilon^*_t}^2}{2 \I_\G(\V_I)}}\right) \nonumber \\
	& = \exp\left({-\frac{\Lambda 2^{t-1} \OPTe \epsilon^2 \I_\G(\V_I) }{2 \I_\G(\V_I) 2^{t-1} \OPTe}}\right) \\
	&= 		\exp\left({- \frac{(2+\frac{2}{3}\epsilon) \ln (\frac{3 t_{\max}}{\delta})\frac{1}{\epsilon^2} 2^{t-1} \OPTe \epsilon^2 \I_\G(\V_I)}{ 2(1+\epsilon/3)\I_\G(\V_I) 2^{t-1} \OPTe}}\right) \nonumber \\
	& \leq \exp\left({-\ln \frac{3 t_{\max}}{\delta}}\right) = \frac{\delta}{3 t_{\max}},
	\end{align}
	which completes the proof of Lemma~\ref{lem:bad2}.
\end{proof}

\begin{Lemma}
	\label{lem:e2e}
	Assume that none of the bad events $B^{(1)}$, $B^{(2)}_t$, $B^{(3)}_t$  ($t =1..t_{\max}$) happen and \emph{\EIS{}} stops with some $\epsilon_t \leq \epsilon$. We have 
	\begin{align}
	&\hat \epsilon_t < \epsilon \text{ and consequently }\\
	&\DE_{t'}(\hat T_k) \leq (1+\epsilon)\DE(\hat T_k)
	\end{align}
\end{Lemma}
\begin{proof}
	Since the bad event $B^{(2)}_t$ does not happen,
	\begin{align}
	\DE_{t'}(\hat T_k) & \leq (1+\hat \epsilon_t)\DE(\hat T_k) \\
	\Leftrightarrow \Cov_{\R'_I}(\hat T_k) & \leq (1+\hat \epsilon_t)N_t\frac{\DE(\hat T_k)}{\I_\G(\V_I)}
	\end{align}
	
	When \EIS{} stops with $\epsilon_t \leq \epsilon$, it must satisfy the condition on Line~2 of Alg.~\ref{alg:check},
	\[
	\Cov_{\R'_I}(\hat T_k) \geq \Lambda_1.
	\] 
	Thus, we have
	\begin{align}
	\label{eq:Nt}(1+\hat \epsilon_t)N_t\frac{\DE(\hat T_k)}{\I_\G(\V_I)} \geq \Lambda_1=1+(1+\epsilon)\frac{2+2/3\epsilon}{\epsilon^2} \ln \frac{3t_{\max}}{\delta}
	\end{align}
	From the definition of $\hat \epsilon_t$, it follows that
	\begin{align}
	\label{eq:nt}
	N_t = \frac{2+2/3 \hat \epsilon_t}{\hat\epsilon_t^2} \ln \left(\frac{3t_{\max}}{\delta}\right)\frac{\I_\G(\V_I)}{\DE(\hat T_k)}
	\end{align}
	Substitute the above into (\ref{eq:Nt}) and simplify, we obtain:
	\begin{align}
	&(1+\hat \epsilon_t)\frac{2+2/3 \hat \epsilon_t}{\hat\epsilon_t^2} \ln \left(\frac{3t_{\max}}{\delta}\right)\\
	\geq & (1+\epsilon)\frac{2+2/3\epsilon}{\epsilon^2} \ln \frac{3t_{\max}}{\delta}+1
	\end{align}
	Since the function $(1+x)\frac{2+2/3x}{x^2}$ is a decreasing function for $x>0$, it follows that $\hat \epsilon_t < \epsilon$.
\end{proof}

We now prove the approximation guarantee of \EIS{}.

\begin{proof}[Proof of Theorem~\ref{theo:app}]
	Apply union bound for the bad events in Lemmas \ref{lem:cap} and \ref{lem:bad2}. The probability that at least one of the bad events $B^{(1)}, B^{(2)}_t, B^{(3)}_t$  ($t =1..t_{\max}$) happen is at most
	\begin{align}
	\delta/3 + \left(\delta/(3t_{\max}) + \delta/(3t_{\max})\right) \times t_{\max} \leq \delta
	\end{align}
	
	In other words, the probability that none of the bad events happen will be at least $1-\delta$. Assume that none of the bad events happen, we shall show that the returned $\hat T_k$ is a $(1-1/e-\epsilon)$-approximation solution.

	If \EIS{} stops with $|\R_t| \geq N_{\max}$, $\hat T_k$ is a $(1-1/e-\epsilon)$-approximation solution, since the bad event $B^{(1)}$ does not happen.
	
	Otherwise, \EIS{} stops at some iteration $t$ and $\epsilon_t \leq \epsilon$. We use contradiction method.  
	Assume that 
	\begin{align}
	\label{eq:contrad}
	\DE(\hat T_k) < (1-1/e-\epsilon) \OPTe.
	\end{align}
	The proof will continue in the following order
	\begin{enumerate}[label=(\Alph*)]
		\item  $\DE(\hat T_k) \geq (1-1/e-\epsilon_t') \OPTe$\\ where $\epsilon'_t = (\epsilon_1 + \hat\epsilon_t + \epsilon_1 \hat\epsilon_t)(1-1/e-\epsilon) + (1-1/e)\epsilon^*_t$.
		\item  $\hat \epsilon_t \leq \epsilon_2$ and  $\epsilon^*_t \leq \epsilon_3$.
		\item  $\epsilon_t' \leq \epsilon_t \leq \epsilon \Rightarrow \DE(\hat T_k) \geq (1-\frac{1}{e}-\epsilon) \OPTe$ (\emph{contradiction}).
	\end{enumerate}

	\emph{Proof of (A)}. Since the bad events $B^{(2)}_t$ and $B^{(3)}_t$ do not happen, we have
	\begin{align}
	\label{eq:ici} \DE_{t'}(\hat T_k) &\leq (1+\hat \epsilon_t)\DE(\hat T_k), \text{and}\\
	\label{eq:iho} \DE_t(T^*_k) &\leq (1-\epsilon_t^*) \OPTe.
	\end{align}
	Since $\epsilon_1 \leftarrow \Cov_{\R_t} (\hat T_k)/\Cov_{\R'_t}(\hat T_k) - 1  = \DE_t(\hat T_k)/\DE_{t'}(\hat T_k) - 1$, it follows from (\ref{eq:ici}) that
	\begin{align}
	\DE_t(\hat T_k) =  (1+\epsilon_1)\DE_{t'}(\hat T_k) \leq (1+\epsilon_1)(1 + \hat\epsilon_t) \DE(\hat T_k) \nonumber
	\end{align}	
	Expand the right hand side and apply (\ref{eq:contrad}), we obtain	
	\begin{align}
	\DE(\hat T_k) &\geq \DE_t(\hat T_k) - (\epsilon_1 + \hat\epsilon_t + \epsilon_1 \hat\epsilon_t) \DE(\hat T_k) \nonumber \\
	& \geq \DE_t(\hat T_k) - (\epsilon_1 + \hat\epsilon_t + \epsilon_1 \hat\epsilon_t) (1-1/e-\epsilon)\OPTe \nonumber
	\end{align}
	Since the \textsf{Greedy} algorithm for the \textsf{Max-Coverage} guarantees an $(1-1/e)$-approximation, $\DE_t(\hat T_k) \geq (1-1/e) \DE_t(T^*_k)$. Thus,
	\begin{align}		
	\DE(\hat T_k) & \geq (1-1/e)\DE_t(T^*_k) \nonumber \\
	& \qquad \qquad - (\epsilon_1 + \hat\epsilon_t + \epsilon_1 \hat\epsilon_t)(1-1/e-\epsilon) \OPTe \nonumber \\
	&\geq (1-1/e)(1-\epsilon^*_t) \OPTe \nonumber \\
	& \qquad \qquad - (\epsilon_1 + \hat\epsilon_t + \epsilon_1 \hat\epsilon_t)(1-1/e-\epsilon) \OPTe  \nonumber \\
	&\geq (1-1/e - \epsilon'_t) \OPTe, \nonumber
	\end{align}
	where $\epsilon'_t = (\epsilon_1 + \hat\epsilon_t + \epsilon_1 \hat\epsilon_t)(1-1/e-\epsilon) + (1-1/e)\epsilon^*_t$.
	
	\emph{Proof of (B)}. We show that $\hat \epsilon_t \leq \epsilon_2$. Due to the computation of $\epsilon_2 \leftarrow \epsilon \sqrt{\frac{|\R_t|(1+\epsilon)}{2^{t-1} \Cov_{\R'_t} (\hat T_k)}}$, we have
	\[
	\frac{1}{\epsilon^2} = \frac{1}{\epsilon_2^2} \frac{|\R'_{t}|}{2^{t-1}} \frac{1+\epsilon}{\Cov_{\R'_{t}}(\hat T_k)} = \frac{1}{\epsilon_2^2} \frac{\I_\G(\V_I)}{2^{t-1}} \frac{1+\epsilon}{\DE_{t'}(\hat T_k)} .
	\]
	Expand the number of \hsaw{} samples in iteration $t$, $N_t = 2^{t-1} \Lambda$, and apply the above equality, we have 
	\begin{align}
	N_t &=  2^{t-1} (2+2/3\epsilon)\frac{1}{\epsilon^2} \ln \frac{3t_{\max}}{\delta} \\
	&=  2^{t-1} (2+2/3\epsilon)\frac{1}{\epsilon_2^2}  \frac{\I_\G(\V_I)}{2^{t-1}} \frac{1+\epsilon}{\DE_{t'}(\hat T_k)}  \ln \frac{3t_{\max}}{\delta} \\
	&= (2+2/3\epsilon)\frac{1}{\epsilon_2^2}   \frac{(1+\epsilon)\I_\G(\V_I)}{\DE_{t'}(\hat T_k)}  \ln \frac{3t_{\max}}{\delta}
	\end{align}
	On the other hand, according to Eq.~(\ref{eq:nt}), we also have,
	\begin{align}
		N_t = \frac{2+2/3 \hat \epsilon_t}{\hat\epsilon_t^2} \ln \left(\frac{3t_{\max}}{\delta}\right)\frac{\I_\G(\V_I)}{\DE(\hat T_k)}.
	\end{align}
	Thus
	\begin{align*}
	& (2+2/3\epsilon)\frac{1}{\epsilon_2^2}   \frac{1+\epsilon}{\DE_{t'}(\hat T_k)} = \frac{2+2/3 \hat \epsilon_t}{\hat\epsilon_t^2} \frac{1}{\DE(\hat T_k)}\\
	\Rightarrow &\frac{\hat \epsilon_t^2}{\epsilon_2^2} = \frac{2+2/3 \hat \epsilon_t}{2+2/3 \epsilon}  \frac{\DE_{t'}(\hat T_k)}{(1+\epsilon)\DE(\hat T_k)} \leq 1
	\end{align*}
	The last step is due to Lemma \ref{lem:e2e}, i.e., $\DE_{t'}(\hat T_k) \leq (1+\epsilon)\DE(\hat T_k)$ and $\hat \epsilon_t \leq \epsilon$. Therefore, $\hat\epsilon_t \leq \epsilon_2$.
	
	We show that $\epsilon^*_t \leq \epsilon_3$. According to the definition of $\epsilon^*_t$ and $\epsilon_3$, we have
	\begin{align*}
		\frac{(\epsilon^*_t)^2}{\epsilon_3^2} &=\frac{\I_\G(\V_I)}{(1+\epsilon/3)2^{t-1} \OPTe} / 
		\frac{\I_\G(\V_I)(1+\epsilon)(1-1/e-\epsilon)}{(1+\epsilon/3)2^{t-1} \DE_{t'} (\hat T_k)}\\
		=& \frac{\DE_{t'} (\hat T_k)}{\OPTe (1+\epsilon)(1-1/e-\epsilon)}
		\leq \frac{\DE_t (\hat T_k)}{\OPTe (1-1/e-\epsilon)}\leq 1
	\end{align*}
	The last two steps follow from Lem. \ref{lem:e2e}, $\DE_{t'}(\hat T_k) \leq (1+\epsilon)\DE(\hat T_k)$ and the assumption 
	(\ref{eq:contrad}), respectively. Thus, $\epsilon_t^* \leq \epsilon_3$.
	
	\emph{Proof of (C)}. Since $1+\epsilon_1 = \hat \DE_t(\hat T_k)/\DE_{t'}(\hat T_k) \geq 0$ and $\epsilon_2 \geq \hat \epsilon_t >0$ and 
	$\epsilon_3 \geq \epsilon^*_t >0$, we have
	\begin{align}
	\epsilon_t' &= (\epsilon_1 + \hat\epsilon_t + \epsilon_1 \hat\epsilon_t)(1-1/e-\epsilon) + (1-1/e)\epsilon^*_t \\
	&= (\epsilon_1 + \hat\epsilon_t (1+\epsilon_1))(1-1/e-\epsilon) + (1-1/e)\epsilon^*_t \\
	&\leq (\epsilon_1 + \epsilon_2 (1+\epsilon_1))(1-1/e-\epsilon) + (1-1/e)\epsilon_3  \\
	&= \epsilon_t \leq \epsilon.
	\end{align}
	This completes the proof.
\end{proof}

\subsection{Node-based Interdiction Algorithms}

\begin{algorithm} \small
	\caption{\textsf{GreedyNode} algorithm for maximum coverage}
	\label{alg:nmax-cover}
	\KwIn{A set $\R_t$ of \hsaw{} samples, $C \subseteq V$ and $k$.}
	\KwOut{An $(1 - 1/e)$-optimal solution $\hat S_k$ on samples.}
	$\hat S_k = \emptyset$;\\
	\For{$i = 1 \emph{\textbf{ to }} k$}{
		$\hat v \leftarrow \arg \max_{v \in C \backslash \hat S_k}(\Cov_{\R_t}(\hat S_k\cup \{e\}) - \Cov_{\R_t}(\hat S_k))$;\\
		Add $\hat v$ to $\hat S_k$;\\
	}
	\textbf{return} $\hat S_k$;
\end{algorithm}
\begin{algorithm} \small
	\caption{\textsf{CheckNode} algorithm for confidence level}
	\label{alg:ncheck}
	\KwIn{$\hat S_k, \R_t, \R'_t, \epsilon, \delta$ and $t$.}
	\KwOut{\textsf{True} if the solution $\hat S_k$ meets the requirement.}
	Compute $\Lambda_1$ by Eq.~\ref{eq:lambda_1};\\
	\If{$\Cov_{\R'_t}(\hat S_k) \geq \Lambda_1$}{
		$\epsilon_1 = \Cov_{\R_t}(\hat S_k)/\Cov_{\R'_t}(\hat S_k) - 1$; \\
		$\epsilon_2 = \epsilon \sqrt{\frac{|\R'_t|(1+\epsilon)}{2^{t-1} \Cov_{\R'_t}(\hat S_k)}}$; $\epsilon_3 = \epsilon \sqrt{\frac{|\R'_t|(1+\epsilon) (1-1/e-\epsilon)}{(1+\epsilon/3)2^{t-1} \Cov_{\R'_t}(\hat S_k)}}$; \\
		$\epsilon_t = (\epsilon_1 + \epsilon_2 + \epsilon_1 \epsilon_2)(1-1/e-\epsilon) + (1-1/e)\epsilon_3$; \\
		\If{$\epsilon_t \leq \epsilon$}{\textbf{return} \textsf{True};}
	}
	\textbf{return} \textsf{False};
\end{algorithm}
\begin{algorithm} \small
	\caption{Node Spread Interdiction Algorithm (\NIS{})}
	\label{alg:nis}
	\KwIn{Graph $\G$, $\V_I$, $p(v), \forall v \in \V_I$, $k$, $C \subseteq V$ and $0 \leq \epsilon, \delta \leq 1$.}
	\KwOut{$\hat S_k$ - An $(1 - 1/e - \epsilon)$-near-optimal solution.}
	Compute $\Lambda$ (Eq.~\ref{eq:lambda}), $N_{\max}$ (Eq.~\ref{eq:nguard}); $t = 0$; \\
	Generate a stream of random samples $R_1, R_2, \dots$ where each $R_j$ is the set of nodes in \hsaw{} sample $h_j$ by Alg.~\ref{alg:par_saw};\\
	\Repeat{$|\R_t|\geq N_{\max}$}{
		$t = t+1$; $\R_t = \{ R_1, \dots, R_{\Lambda 2^{t-1}} \}; \R'_t = \{ R_{\Lambda 2^{t-1} + 1}, \dots, R_{\Lambda 2^{t}} \}$;\\
		$\hat S_k \leftarrow \textsf{GreedyNode}(\R_t, C, k)$;\\
		\If{$\emph{\textsf{CheckNode}}(\hat S_k, \R_t, \R'_t, \epsilon,\delta) = \emph{\textsf{True}}$}{\textbf{return} $\hat S_k$;}
	}
	\textbf{return} $\hat S_k$;\\
\end{algorithm}

\captionsetup{width=0.9\textwidth}
\begin{figure*}[ht]
	\subfloat[Pokec]{
		\includegraphics[width=0.24\linewidth]{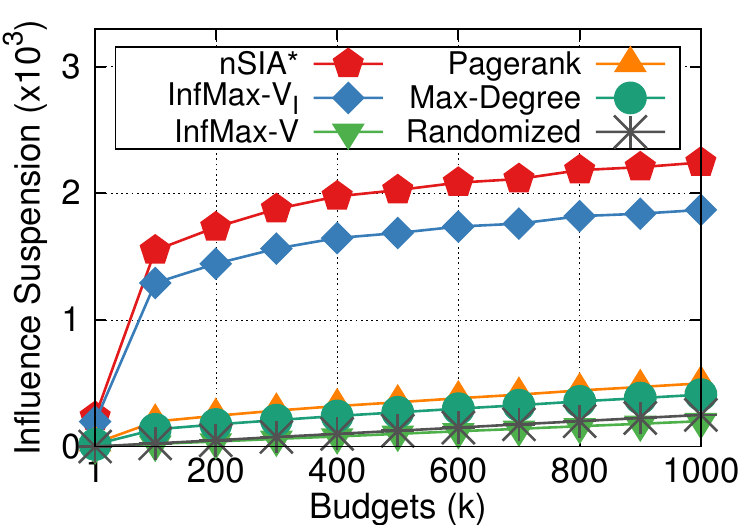}
	}
	\subfloat[Skitter]{
		\includegraphics[width=0.24\linewidth]{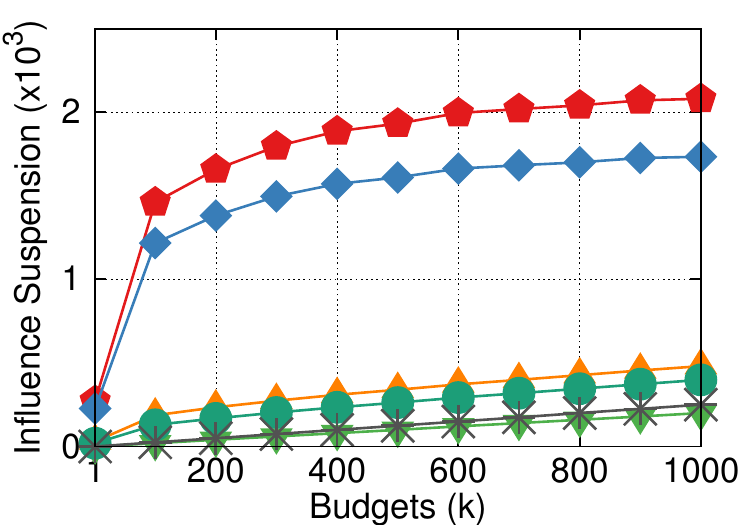}
	}
	\subfloat[LiveJournal]{
		\includegraphics[width=0.24\linewidth]{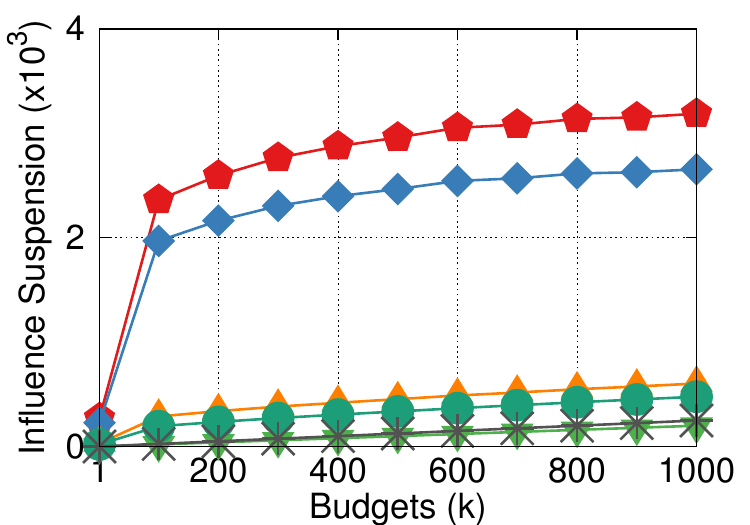}
	}
	\subfloat[Twitter]{
		\includegraphics[width=0.24\linewidth]{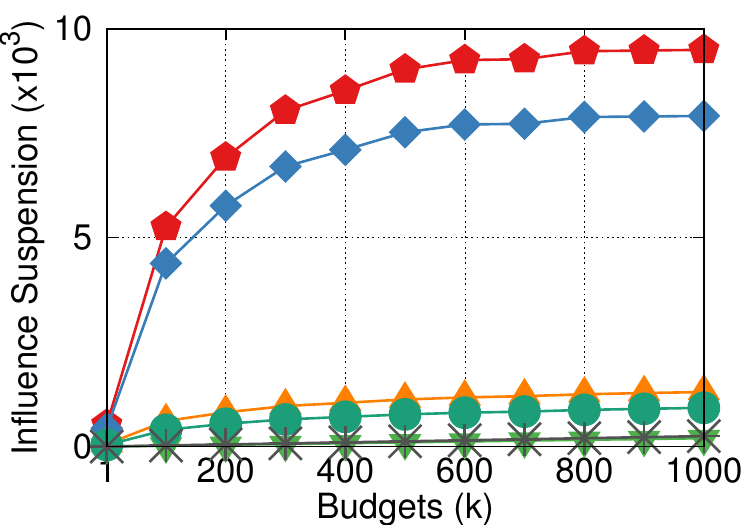}
	}
	\caption{Interdiction Efficiency of different approaches on \NRI{} problem (\NIS{}* refers to general \NIS{} algorithm)}
	\label{fig:node_inf}
\end{figure*}

The maximum number of \hsaw{} samples needed,
\begin{align}
	\label{eq:nguard}
	N_{\max} = (2-\frac{1}{e})^2(2+\frac{2}{3}\epsilon) n \cdot \frac{\ln (6/\delta)+\ln {n \choose k}}{k \epsilon^2},
\end{align}

%

\textbf{Algorithm Description.} Alg.~\ref{alg:eis} uses two subroutines:
\begin{itemize}
	\item[1)] \textsf{Greedy} in Alg.~\ref{alg:nmax-cover} that selects a candidate solution $\hat S_k$ from a set of \hsaw{} samples $\R_t$. This implements the greedy scheme that selects a node with maximum marginal gain and add it to the solution until $k$ nodes have been selected.
	\item[2)] \textsf{Check} in Alg.~\ref{alg:check} that checks the candidate solution $\hat S_k$ if it satisfies the given precision error $\epsilon$. It computes the error bound provided in the current iteration of \NIS{}, i.e. $\epsilon_t$ from $\epsilon_1, \epsilon_2, \epsilon_3$ (Lines~4-6), and compares that with the input $\epsilon$. This algorithm consists of a checking condition (Line~2) that examines the coverage of $\hat S_k$ on second independent set $\R'_t$ from $\R_t$ (used in \textsf{Greedy} to find $\hat S_k$) of \hsaw{} samples. The computations of $\epsilon_1, \epsilon_2, \epsilon_3$ are to guarantee the estimation quality of $\hat S_k$ and the optimal solution $S^*_k$.
\end{itemize}

The main algorithm in Alg.~\ref{alg:eis} first computes $\Lambda$ and the upper-bound on neccessary \hsaw{} samples in Line~1. Then, it enters a loop of at most $t_{\max}$ iterations. In each iteration, \NIS{} uses the set $\R_t$ of first $\Lambda 2^{t-1}$  \hsaw{} samples to find a candidate solution $\hat S_k$ by the \textsf{Greedy} algorithm (Alg.~\ref{alg:max-cover}). Afterwards, it checks the quality of $\hat S_k$ by the \textsf{Check} procedure (Alg.~\ref{alg:check}). If the \textsf{Check} returns \textsf{True} meaning that $\hat S_k$ meets the error requirement $\epsilon$ with high probability, $\hat S_k$ is returned as the final solution. 

In cases when \textsf{Check} algorithm fails to verify the candidate solution $\hat S_k$ after $t_{\max}$ iteations, \NIS{} will be terminated by the guarding condition $|\R_t| \geq N_{\max}$ (Line~9). 

\newpage
\section{Performance on Node-based Spread Interdiction}

The results of comparing the solution quality, i.e., influence suspension, of the algorithms on the four larger network datasets, e.g., Pokec, Skitter, LiveJournal and Twitter, are presented in Fig.~\ref{fig:node_inf} for \NRI{}. Across all four datasets, we observe that \NIS{} significantly outperforms the other methods with widening margins when $k$ increases. For example, on the largest Twitter network, \NIS{} is about 20\% better than the runner-up \textsf{InfMax-$V_I$} and 10 times better than the rest.

\begin{table*}[htb] \centering
	\begin{tabular}{l | rrrrrrrrrrr}
		\toprule
		\multirow{2}{*}{\textbf{Methods}} & \multicolumn{5}{c}{\textbf{Suspect Selection Ratio (SSR)}} && \multicolumn{5}{c}{\textbf{Interdiction Cost}} \\
		\cline{2-6}\cline{8-12}
		& 200 & 400 & 600 & 800 & 1000 && 200 & 400 & 600 & 800 & 1000 \\
		\midrule
		\textsf{\NIS{}} & 0.99 & 0.96 & 0.94 & 0.92 & 0.85  && \textbf{303.55} & \textbf{541.73} & \textbf{720.73} & 1070.50 & \textbf{1204.53} \\
		\textsf{InfMax}-$V_I$ & 1.00 & 1.00 & 1.00 & 1.00 & 1.00  && 333.90 & 600.60 & 825.75 & \textbf{1020.00} & 1246.03 \\
		\textsf{InfMax}-$V$ & 0.00 & 0.00 & 0.01 & 0.01 & 0.01  && 754.30 & 1302.85 & 1876.78 & 2427.58 & 2938.95 \\
		\textsf{Pagerank} & 0.00 & 0.00 & 0.00 & 0.00 & 0.00  && 727.29 & 1562.06 & 2551.15 & 3301.35 & 4229.84 \\
		\textsf{Max-Degree} & 0.00 & 0.00 & 0.00 & 0.00 & 0.00  && 820.03 & 1711.88 & 2660.15 & 3561.95 & 4505.40 \\
		\textsf{Randomized} & 0.00 & 0.01 & 0.02 & 0.02 & 0.04  && 486.85 & 962.78 & 1424.65 & 1885.00 & 2364.33 \\
		\bottomrule
	\end{tabular}
	\caption{Interruption levels from removing nodes in the networks}
	\label{tab:precision}
\end{table*}

\section{Analyzing Nodes Selected for Removal}
We aim at analyzing which kind of nodes, i.e., those in $\V_I$ or popular nodes, selected by different algorithms.

\textbf{Suspect Selection Ratio.} We first analyze the solutions using \textsf{Suspect Selection Ratio (SSR)} which is the ratio of the number of suspected nodes selected to the budget $k$. We run the algorithms on all five datasets and take the average. Our results are presented in Table~\ref{tab:precision}. We see that the majority of nodes selected by \NIS{} are suspected while the other methods except \textsf{InfMax-$V_I$} rarely select these nodes. \textsf{InfMax-$V_I$} restricts its selection in $V_I$ and thus, always has value of 1. The small faction of nodes selected by \NIS{} not in $V_I$ are possibly border nodes between $V_I$ and the outside.

\textbf{Interdiction Cost Analysis.} We measure the cost of removing a set $S$ of nodes by the following function.
\begin{align}
	Cost (S) = \sum_{v \in S} (1-p(v)) \log(d_{in}(v)+1)
\end{align}
The $Cost(S)$ function consists of two parts: 1) probability of node being suspected and 2) the popularity of that node in the network implied by its in-degree. Interdicting nodes with higher probability or less popular results in smaller cost indicating less interruption to the operations of network.

The Interdiction costs of the solutions returned by different algorithms are shown in Table~\ref{tab:precision}. We observe that \NIS{} introduces the least cost, even less than that of \textsf{InfMax}-$V_I$ possibly because \textsf{InfMax}-$V_I$ ignores the probabilities of nodes being suspected. The other methods present huge cost to networks since they target high influence nodes.
\end{document}